\documentclass{article}

\usepackage{latexsym}
\usepackage{amsmath, amsfonts, amsthm}
\usepackage{epsfig}
\usepackage{stmaryrd}
\usepackage{multicol}
\usepackage{pdfsync}
\usepackage{dsfont}
\usepackage{comment}
\usepackage{algorithmic}
\usepackage{algorithm}
\usepackage{caption}
\usepackage{psfrag}
\usepackage{xcolor}
\usepackage{graphics}

\usepackage{mathtools}

\newcommand{\norm}[1]{\left\lVert #1 \right\rVert}

\definecolor{OliveGreen}{rgb}{0,0.6,0}

\newtheorem{theorem}{Theorem}[section]

\theoremstyle{definition}

\theoremstyle{remark}
\newtheorem{remark}[theorem]{Remark}
\newtheorem{example}[theorem]{Example}
\newtheorem{condition}[theorem]{Condition}
\numberwithin{equation}{section}

\usepackage{setspace}

\title{DGM: A deep learning algorithm for solving partial differential equations}

\author{Justin Sirignano\footnote{University of Illinois at Urbana Champaign, Urbana, E-mail: jasirign@illinois.edu} \phantom{.}  and Konstantinos Spiliopoulos\footnote{Department of Mathematics and Statistics, Boston University, Boston, E-mail: kspiliop@math.bu.edu} \thanks{The authors
thank seminar participants at the JP Morgan Machine Learning and AI Forum seminar, the Imperial College London Applied Mathematics and Mathematical Physics seminar, the Department of Applied Mathematics at the University of Colorado Boulder, Princeton University, and Northwestern University for their comments. The authors would also like to thank participants at the 2017 INFORMS Applied Probability Conference, the 2017 Greek Stochastics Conference, and the 2018 SIAM Annual Meeting for their comments.} \thanks{Research of K.S. supported in part by the National Science Foundation (DMS 1550918). Computations for this paper were performed using the Blue Waters supercomputer grant ``Distributed Learning with Neural Networks". }  \\
}

\date{\today}

\begin{document}

\maketitle

\begin{abstract}High-dimensional PDEs have been a longstanding computational challenge. We propose to solve high-dimensional PDEs by approximating the solution with a deep neural network which is trained to satisfy the differential operator, initial condition, and boundary conditions. Our algorithm is meshfree, which is key since meshes become infeasible in higher dimensions. Instead of forming a mesh, the neural network is trained on batches of randomly sampled time and space points. The algorithm is tested on a class of high-dimensional free boundary PDEs, which we are able to accurately solve in up to $200$ dimensions. The algorithm is also tested on a high-dimensional Hamilton-Jacobi-Bellman PDE \textcolor{black}{and Burgers' equation.} The deep learning algorithm approximates the general solution to the Burgers' equation for a continuum of different boundary conditions and physical conditions (which can be viewed as a high-dimensional space). We call the algorithm a ``Deep Galerkin Method (DGM)" since it is similar in spirit to Galerkin methods, with the solution approximated by a neural network instead of a linear combination of basis functions. In addition, we prove a theorem regarding the approximation power of neural networks for \textcolor{black}{a class of quasilinear parabolic PDEs}.
\end{abstract}

\section{Deep learning and high-dimensional PDEs}
High-dimensional partial differential equations (PDEs) are used in physics, engineering, and finance. Their numerical solution has been a longstanding challenge. Finite difference methods become infeasible in higher dimensions due to the explosion in the number of grid points and the demand for reduced time step size. If there are $d$ space dimensions and $1$ time dimension, the mesh is of size $\mathcal{O}^{d+1}$. \textcolor{black}{This quickly becomes computationally intractable when the dimension $d$ becomes even moderately large. We propose to solve high-dimensional PDEs using a meshfree deep learning algorithm.} The method is similar in spirit to the Galerkin method, but with several key changes using ideas from machine learning. \textcolor{black}{The Galerkin method is a widely-used computational method which seeks a reduced-form solution to a PDE as a linear combination of basis functions.} The deep learning algorithm, or ``Deep Galerkin Method" (DGM), uses a deep neural network instead of a linear combination of basis functions. The deep neural network is trained to satisfy the differential operator, initial condition, and boundary conditions using stochastic gradient descent at randomly sampled spatial points. By randomly sampling spatial points, we avoid the need to form a mesh (which is infeasible in higher dimensions) and instead convert the PDE problem into a machine learning problem.

DGM is a natural merger of Galerkin methods and machine learning. The algorithm in principle is straightforward; see Section \ref{Algorithm}. Promising numerical results are presented later in Section \ref{NumericalAnalysis} for a class of high-dimensional free boundary PDEs. We also accurately solve a high-dimensional Hamilton-Jacobi-Bellman PDE in Section \ref{HJBsection} \textcolor{black}{and Burger's equation in Section \ref{BurgerEquation}}. DGM converts the computational cost of finite difference to a more convenient form: instead of a huge mesh of $\mathcal{O}^{d+1}$ (which is infeasible to handle), many batches of random spatial points are generated. Although the total number of spatial points could be vast, the algorithm can process the spatial points sequentially without harming the convergence rate.

\textcolor{black}{Deep learning has revolutionized fields such as image, text, and speech recognition. These fields require statistical approaches which can model nonlinear functions of high-dimensional inputs. Deep learning, which uses multi-layer neural networks (i.e., ``deep neural networks"), has
proven very effective in practice for such tasks. A multi-layer neural network is essentially a ``stack" of nonlinear operations where each operation is prescribed by certain parameters that must be estimated from data. Performance in practice can strongly depend upon the specific form of the neural network architecture and the training algorithms which are used. The design of neural network architectures and training methods has been the focus of intense research over the past decade. Given the success of deep learning, there is also growing interest in applying it to a range of other areas in science and engineering (see Section \ref{RelevantLiterature} for some examples).}

Evaluating the accuracy of the deep learning algorithm is not straightforward. PDEs with semi-analytic solutions may not be sufficiently challenging. (After all, the semi-analytic solution exists since the PDE can be transformed into a lower-dimensional equation.) It cannot be benchmarked against traditional finite difference (which fails in high dimensions). We test the deep learning algorithm on a class of high-dimensional free boundary PDEs which have the special property that error bounds can be calculated for any approximate solution. This provides a unique opportunity to evaluate the accuracy of the deep learning algorithm on a class of high-dimensional PDEs \emph{with no semi-analytic solutions}.

This class of high-dimensional free boundary PDEs also has important applications in finance, where it used to price American options. An American option is a financial derivative on a portfolio of stocks. The number of space dimensions in the PDE equals the number of stocks in the portfolio. Financial institutions are interested in pricing options on portfolios ranging from dozens to even hundreds of stocks \cite{Reisinger}. Therefore, there is a significant need for numerical methods to accurately solve high-dimensional free boundary PDEs.

We also test the deep learning algorithm on a high-dimensional Hamilton-Jacobi-Bellman PDE with accurate results. We consider a high-dimensional Hamilton-Jacobi-Bellman PDE motivated by the problem of optimally controlling a stochastic heat equation.

\textcolor{black}{Finally, it is often of interest to find the solution of a PDE over a range of problem setups (e.g., different physical conditions and boundary conditions). For example, this may be useful for the design of engineering systems or uncertainty quantification. The problem setup space may be high-dimensional and therefore may require solving many PDEs for many different problem setups, which can be computationally expensive. We use our deep learning algorithm to approximate the general solution to the Burgers' equation for different boundary conditions, initial conditions, and physical conditions.}

\textcolor{black}{In the remainder of the Introduction, we provide an overview of our results regarding the approximation power of neural networks for quasilinear parabolic PDEs (Section \ref{IntroApproxPower}), and relevant literature (Section \ref{RelevantLiterature}).} \textcolor{black}{The deep learning algorithm for solving PDEs is presented in Section \ref{Algorithm}. An efficient scheme for evaluating the diffusion operator is developed in Section \ref{ModifiedAlgorithm}.  Numerical analysis of the algorithm is presented in Sections \ref{NumericalAnalysis}, \ref{HJBsection}, and \ref{BurgerEquation}. We implement and test the algorithm on a class of high-dimensional free boundary PDEs in up to 200 dimensions. The theorem and proof for the approximation of PDE solutions with neural networks is presented in Section \ref{ApproximationProof}. Conclusions are in Section \ref{Conclusion}.  \textcolor{black}{For readability purposes, proofs from Section \ref{ApproximationProof} have been collected in  Appendix \ref{S:Proofs}.}}
\subsection{Approximation Power of Neural Networks for PDEs} \label{IntroApproxPower}
We also prove a theorem regarding the approximation power of neural networks for \textcolor{black}{a class of quasilinear parabolic  PDEs}. Consider the potentially nonlinear PDE
\begin{eqnarray}
&& \partial_t u(t,x)  + \mathcal{L} u(t,x) = 0, \phantom{.......} (t,x) \in [0,T] \times \Omega \notag \\
&& u(0,x) = u_0(x), \phantom{........} x \in \Omega \notag \\
&& u(t,x) = g(t,x), \phantom{........} x \in [0,T] \times \partial \Omega,
\label{GeneralPDE}
\end{eqnarray}
\textcolor{black}{where $\partial \Omega$ is the boundary of the domain $\Omega$}. The solution $u(t,x)$ is of course unknown, but an approximate solution $f(t,x)$ can be found by minimizing the L$^2$ error
\begin{align}
J( f )&=\norm{ \partial_t f + \mathcal{L}f }^{2}_{2, [0,T] \times \Omega } +   \norm{  f -g}^{2}_{2, [0,T]  \times \partial \Omega} + \norm{ f (0, \cdot) - u_0 }^{2}_{2, \Omega}. \notag
\end{align}
The error function $J(f)$ measures how well the approximate solution $f$ satisfies the differential operator, boundary condition, and initial condition. Note that no knowledge of the actual solution $u$ is assumed; $J(f)$ can be directly calculated from the PDE (\ref{GeneralPDE}) for any approximation $f$. The goal is to construct functions $f$ for which  $J(f)$  is as close to  $0$ as possible. Define $\mathfrak{C}^{n}$ as the class of neural networks \textcolor{black}{with a single hidden layer and $n$ hidden units.}\footnote{\textcolor{black}{A neural network with a single hidden layer and $n$ hidden units is a function of the form $\mathfrak{C}^{n}=\left\{h(t,x):\mathbb{R}^{1+d}\mapsto\mathbb{R}: h(t,x)=\sum_{i=1}^{n}\beta_{i}\psi\left(\alpha_{1,i}t+\sum_{j=1}^{d}\alpha_{j,i}x_{j}+c_{j}\right)\right\}$ where $\Psi: \mathbb{R} \rightarrow \mathbb{R}$ is a nonlinear ``activation" function such as a sigmoid or tanh function.}} Let $f^{n}$ be a neural network with $n$ hidden units which minimizes $J(f)$. We prove that, under certain conditions,
\begin{eqnarray}
&& \text{there exists }f^{n}\in \mathfrak{C}^{n} \text{ such that }  J(f^{n})\rightarrow 0, \text{ as }n\rightarrow \infty, \text{ and} \notag \\
&& f^{n} \rightarrow u \phantom{....} \textrm{as} \phantom{....} n \rightarrow \infty, \notag
\end{eqnarray}
strongly in, $L^{\rho}([0,T] \times \Omega)$, with $\rho<2$, for \textcolor{black}{a class of quasilinear parabolic PDEs}; see subsection \ref{ConvergenceProof} and Theorem \ref{T:MainConvTheorem} therein for the precise statement. That is, the neural network will converge in $L^{\rho}, \rho<2$ to the solution of the PDE \textcolor{black}{as the number of hidden units tends to infinity.}  The precise statement of the theorem and its proof are presented in Section \ref{ApproximationProof}. The proof requires the joint analysis of the approximation power of neural networks as well as the continuity properties of partial differential equations. Note that $J(f^{n}) \rightarrow 0$ does not necessarily imply that $f^{n} \rightarrow u$, given that we only have $L^{2}$ control on the approximation error. First, we prove that $J(f^{n}) \rightarrow 0$ as $n \rightarrow \infty$. We then establish that each neural network $\{ f^{n} \}_{n=1}^{\infty}$ satisfies a PDE with a source term $h^n(t,x)$.   \textcolor{black}{We are then able to prove, under certain conditions, the convergence of $f^{n} \rightarrow u$ as $n \rightarrow \infty$ in $L^{\rho}([0,T] \times \Omega)$, for $\rho<2$,    using  the smoothness of the neural network approximations and compactness arguments.}

Theorem \ref{T:MainConvTheorem} establishes the approximation power of neural networks for solving PDEs (at least within \textcolor{black}{a class of quasilinear parabolic PDEs}); however, directly minimizing $J(f)$ is not computationally tractable since it involves high-dimensional integrals. The DGM algorithm minimizes $J(f)$ using a meshfree approach; see Section \ref{Algorithm}.

\subsection{Relevant Literature} \label{RelevantLiterature}

Solving PDEs with a neural network as an approximation is a natural idea, and has been considered in various forms previously. \cite{Lagaris}, \cite{Likas}, \cite{Rudd},  \cite{Lee}, and \cite{Malek} propose to use neural networks to solve PDEs and ODEs.  These papers estimate neural network solutions on an a priori fixed mesh. This paper proposes using \emph{deep} neural networks and is \emph{meshfree}, which is key to solving high-dimensional PDEs.

In particular, this paper explores several new innovations. First, we focus on high-dimensional PDEs and apply deep learning advances of the past decade to this problem (deep neural networks instead of shallow neural networks, improved optimization methods for neural networks, etc.). Algorithms for high-dimensional free boundary PDEs are developed, efficiently implemented, and tested. In particular, we develop an iterative method to address the free boundary. Secondly, to avoid ever forming a mesh, we sample a sequence of random spatial points. This produces a meshfree method, which is essential for high-dimensional PDEs. Thirdly, the algorithm incorporates a new computational scheme for the efficient computation of neural network gradients arising from the second derivatives of high-dimensional PDEs.

Recently, \cite{Karniadakis1,Karniadakis2} develop physics informed deep learning models. They estimate deep neural network models which merge data observations with PDE models. This allows for the estimation of  physical models from limited data by leveraging a priori knowledge that the physical dynamics should obey a class of PDEs. Their approach solves PDEs in one and two spatial dimensions using deep neural networks. \cite{Ling} uses a deep neural network to model the Reynolds stresses in a Reynolds-averaged Navier-Stokes (RANS) model. RANS is a reduced-order model for turbulence in fluid dynamics. \cite{Weinan, Jentzen} have also recently developed a scheme for solving a class of quasilinear PDEs which can be represented as forward-backward stochastic differential equations (FBSDEs) and \cite{Masaaki} further develops the algorithm. The algorithm developed in \cite{Weinan,Jentzen,Masaaki} focuses on computing the value of the PDE solution at a single point. \textcolor{black}{The algorithm that we present here is  different; in particular, it does not rely on the availability of FBSDE representations and yields the entire solution of the PDE across all time and space. In addition, the deep neural network architecture that we use, which is different from the ones used in \cite{Weinan, Jentzen}, seems to be able to recover accurately the entire solution (at least for the equations that we studied). } \cite{Tompson} use a convolutional neural network to solve a large sparse linear system which is required in the numerical solution of the Navier-Stokes PDE. In addition, \cite{Chaudhari} has recently developed a novel partial differential equation approach to optimize deep neural networks.

\cite{LongstaffSchwartz} developed an algorithm for the solution of a discrete-time version of a class of free boundary PDEs. Their algorithm, commonly called the ``Longstaff-Schwartz method", uses dynamic programming and approximates the solution using a separate function approximator at each discrete time (typically a linear combination of basis functions). Our algorithm directly solves the PDE, and uses a single function approximator for all space and all time. The Longstaff-Schwartz algorithm has been further analyzed by \cite{ChrisRogers}, \cite{Haugh}, and others.  Sparse grid methods have also been used to solve high-dimensional PDEs; see \cite{Reisinger}, \cite{Reisinger2}, \cite{Bungartz1}, \cite{Bungartz2}, and \cite{Griebel}.

\textcolor{black}{In regards to general results on the approximation power of neural networks we refer the interested reader to classical works \cite{Cybenko, HornikEtAl,Hornik, Pinkus} and we also mention the recent work by \cite{Petersen2017}, where the authors study the necessary and sufficient complexity of ReLU neural networks that is required for approximating classifier functions in the mean square sense.}

\section{Algorithm} \label{Algorithm}
Consider a parabolic PDE with $d$ spatial dimensions:
\begin{eqnarray}
&\phantom{.}& \frac{\partial u}{\partial t}(t,x) +  \mathcal{L} u(t,x)  = 0 , \phantom{....} (t,x) \in [0,T] \times \Omega, \notag \\
& \phantom{.}& u(t=0,x )  = u_0(x), \notag \\
& \phantom{.} & u(t,x) = g(t, x),  \phantom{....} x \in \partial \Omega,
\label{ParabolicPDE}
\end{eqnarray}
where $x \in \Omega \subset \mathbb{R}^d$. The DGM algorithm approximates $u(t,x)$ with a deep neural network $f(t,x; \theta)$ where $\theta \in \mathbb{R}^K$ are the neural network's parameters. Note that the differential operators $\frac{\partial f}{\partial t}(t, x; \theta)$ and $\mathcal{L} f(t,x; \theta)$ can be calculated analytically.  Construct the objective function:
\begin{eqnarray*}
J(f) =   \norm{  \frac{\partial f}{\partial t}(t,x; \theta) +  \mathcal{L} f(t,x; \theta)}_{[0,T] \times \Omega, \nu_1}^2 +   \norm{  f(t,x; \theta ) - g(t,x) }_{ [0,T] \times \partial \Omega, \nu_2}^2 + \norm{ f(0, x; \theta) - u_0(x) }_{ \Omega, \nu_3}^2.
\end{eqnarray*}
Here, $\norm{ f(y) }^{2}_{\mathcal{Y}, \nu} = \int_{\mathcal{Y}} \left|f(y)\right|^2  \nu(y) dy$ where $\nu(y)$ is a positive probability density on $y \in \mathcal{Y}$.  $J(f)$ measures how well the function $f(t,x; \theta)$ satisfies the PDE differential operator, boundary conditions, and initial condition.  If $J(f) = 0$, then $f(t,x; \theta)$ is a solution to the PDE (\ref{ParabolicPDE}).

The goal is to find a set of parameters $\theta$ such that the function $f(t,x; \theta)$ minimizes the error $J(f)$.  If the error $J(f)$ is small, then $f(t,x; \theta)$ will closely satisfy the PDE differential operator, boundary conditions, and initial condition. Therefore, a $\theta$ which minimizes $J(f(\cdot; \theta))$ produces a reduced-form model $f(t,x; \theta)$ which approximates the PDE solution $u(t,x)$.

Estimating $\theta$ by directly minimizing $J(f)$ is infeasible when the dimension $d$ is large since the integral over $\Omega$ is computationally intractable.  However, borrowing a machine learning approach, one can instead minimize $J(f)$ using stochastic gradient descent on a sequence of time and space points \emph{drawn at random} from $\Omega$ and $\partial \Omega$. This avoids ever forming a mesh.

The DGM algorithm is:
\begin{enumerate} \label{DGMalgorithM}
\item Generate random points $(t_n, x_n)$ from $[0,T] \times \Omega$ and $(\tau_n, z_n)$ from $[0,T] \times \partial \Omega$ according to respective probability densities $\nu_1$ and $\nu_2$.  Also, draw the random point $w_n$ from $\Omega$ with probability density $\nu_3$.
\item Calculate the squared error $G(\theta_n ,s_n)$ at the randomly sampled points $s_n = \{ (t_n, x_n), ( \tau_n, z_n), w_n \}$ where:
\begin{eqnarray*}
G(\theta_n , s_n ) =   \bigg{(} \frac{\partial f}{\partial t}(t_n ,x_n; \theta_n) +  \mathcal{L} f(t_n,x_n; \theta_n) \bigg{)}^2 +   \bigg{(} f(\tau_n ,z_n ; \theta_n ) - g(\tau_n,z_n) \bigg{)}^2 + \bigg{(}  f(0, w_n; \theta_n) - u_0(w_n) \bigg{)}^2. \notag
\end{eqnarray*}
\item Take a descent step at the random point $s_n$:
\begin{eqnarray*}
\theta_{n+1} = \theta_n - \alpha_n \nabla_{\theta} G(\theta_n, s_n)  \notag
\end{eqnarray*}
\item Repeat until convergence criterion is satisfied.
\end{enumerate}
\textcolor{black}{The ``learning rate" $\alpha_n$ decreases with $n$.} The steps $ \nabla_{\theta} G(\theta_n, s_n) $ are unbiased estimates of $\nabla_{\theta} J( f(\cdot;\theta_n) )$:
\begin{eqnarray*}
\mathbb{E}\big{[} \nabla_{\theta} G(\theta_n, s_n) \big{|} \theta_n  \big{]}  = \nabla_{\theta} J(f(\cdot; \theta_n)).
\end{eqnarray*}
 \textcolor{black}{Therefore, the stochastic gradient descent algorithm will on average take steps in a \emph{descent direction} for the objective function $J$. A descent direction means that the objective function decreases after an iteration (i.e., $J( f(\cdot;\theta_{n+1}) ) < J( f(\cdot;\theta_n) )$ ), and $\theta_{n+1}$ is therefore a better parameter estimate than $\theta_{n}$.}

Under (relatively mild) technical conditions (see \cite{BertsekasThitsiklis2000}), the algorithm $\theta_n$ will converge to a critical point of the objective function $J(f(\cdot;\theta))$ as $n \rightarrow \infty$:
\begin{eqnarray*}
\lim_{n \rightarrow \infty} \norm{ \nabla_{\theta} J( f(\cdot;\theta_n) ) }  = 0.
\end{eqnarray*}

It's important to note that $\theta_n$ may only converge to a local minimum when $f(t, x; \theta)$ is non-convex. This is generally true for non-convex optimization and is not specific to this paper's algorithm. \textcolor{black}{In particular, deep neural networks are non-convex. Therefore, it is well known that stochastic gradient descent may only converge to a local minimum (and not a global minimum) for a neural network. Nevertheless, stochastic gradient descent has proven very effective in practice and is the fundamental building block of nearly all approaches for training deep learning models.}

\section{A Monte Carlo Method for Fast Computation of Second Derivatives} \label{ModifiedAlgorithm}
This section describes a modified algorithm which may be more computationally efficient in some cases. The term $\mathcal{L} f(t,x; \theta)$ contains second derivatives $\frac{\partial^2 f}{\partial x_i x_j} (t,x; \theta)$ which may be expensive to compute in higher dimensions. For instance, $20,000$ second derivatives must be calculated in $d=200$ dimensions.

\textcolor{black}{The complicated architectures of neural networks can make it computationally costly to calculate the second derivatives (for example, see the neural network architecture (\ref{ArchitectureNeuralNetwork})). The computational cost for calculating second derivatives (in both total arithmetic operations and memory) is $\mathcal{O}(d^2 \times N)$ where $d$ is the spatial dimension of $x$ and $N$ is the batch size. In comparison, the computational cost for calculating first derivatives is $\mathcal{O}(d \times N)$. The cost associated with the second derivatives is further increased since we actually need the third-order derivatives $\nabla_{\theta} \frac{\partial^2 f}{\partial x^2}(t,x; \theta)$ for the stochastic gradient descent algorithm. Instead of directly calculating these second derivatives, we approximate the second derivatives using a Monte Carlo method.}

Suppose the sum of the second derivatives in $\mathcal{L} f(t,x,; \theta)$ is of the form $\frac{1}{2} \sum_{i,j=1}^d \rho_{i,j} \sigma_{i}(x) \sigma_j(x)  \frac{\partial^2 f}{\partial x_i x_j} (t,x; \theta)$, assume $[\rho_{i,j}]_{i,j=1}^d$ is a positive definite matrix, and define $\sigma(x) = \bigg{(} \sigma_1(x), \ldots, \sigma_d(x) \bigg{)}$. \textcolor{black}{For example, such PDEs arise when considering expectations of functions of stochastic differential equations, where the $\sigma(x)$ represents the diffusion coefficient. See equation (\ref{SDEx}) and the corresponding discussion.} A generalization of the algorithm in this section to second derivatives with nonlinear coefficient dependence on $u(t,x)$ is also possible. Then,
\begin{eqnarray}
\sum_{i,j=1}^d \rho_{i,j} \sigma_{i}(x) \sigma_j(x)  \frac{\partial^2 f}{\partial x_i x_j} (t,x; \theta) =  \lim_{\Delta \rightarrow 0} \mathbb{E} \bigg{[} \sum_{i=1}^d  \frac{ \frac{\partial f}{\partial x_i}(t, x + \sigma(x)W_{\Delta}; \theta ) - \frac{\partial f}{\partial x_i}(t,x; \theta)  }{\Delta}  \sigma_i(x)  W_{\Delta}^i  \bigg{]} ,
\label{LimIto}
\end{eqnarray}
where $W_t \in \mathbb{R}^d$ is a Brownian motion \textcolor{black}{and $\Delta \in \mathbb{R}_{+}$ is the step-size}. \textcolor{black}{The convergence rate for (\ref{LimIto}) is $\mathcal{O}(\sqrt{ \Delta})$.}\footnote{\textcolor{black}{Let $f$ be a three-times differentiable function in $x$ with bounded third-order derivatives in $x$. Then, it directly follows from a Taylor expansion that $\bigg{|} \sum_{i,j=1}^d \rho_{i,j} \sigma_{i}(x) \sigma_j(x)  \frac{\partial^2 f}{\partial x_i x_j} (t,x; \theta) -  \mathbb{E} \bigg{[} \sum_{i=1}^d  \frac{ \frac{\partial f}{\partial x_i}(t, x + \sigma(x)W_{\Delta}; \theta ) - \frac{\partial f}{\partial x_i}(t,x; \theta)  }{\Delta}  \sigma_i(x)  W_{\Delta}^i  \bigg{]}  \bigg{|} \leq C(x) \sqrt{ \Delta} $. The constant $C(x)$ depends upon $\rho, f_{xxx}(t,x; \theta)$ and $\sigma(x)$.}} Define:
\begin{eqnarray*}
G_1(\theta_n , s_n ) &\textcolor{black}{\vcentcolon =}&   \bigg{(} \frac{\partial f}{\partial t}(t_n ,x_n; \theta_n) +  \mathcal{L} f(t_n,x_n; \theta_n) \bigg{)}^2, \notag \\
G_2(\theta_n , s_n ) &\textcolor{black}{\vcentcolon =}&   \bigg{(} f(\tau_n ,z_n ; \theta_n ) - g(\tau_n ,z_n) \bigg{)}^2 , \notag \\
G_3(\theta_n , s_n )  &\textcolor{black}{\vcentcolon =}&  \bigg{(}  f(0, w_n; \theta_n) - u_0(w_n) \bigg{)}^2, \notag \\
G(\theta_n , s_n )  &\textcolor{black}{\vcentcolon =}& G_1(\theta_n , s_n )  + G_2(\theta_n , s_n )  + G_3(\theta_n , s_n ) .
\end{eqnarray*}
The DGM algorithm use the gradient $\nabla_{\theta} G_1 (\theta_n , s_n ) $, which requires the calculation of the second derivative terms in $ \mathcal{L} f(t_n,x_n; \theta_n)$.  Define the first derivative operators as
\[
 \mathcal{L}_1 f(t_n,x_n; \theta_n)   \textcolor{black}{ \vcentcolon =}   \mathcal{L} f(t_n,x_n; \theta_n) - \frac{1}{2} \sum_{i,j=1}^d \rho_{i,j} \sigma_{i}(x_n) \sigma_j(x_n)  \frac{\partial^2 f}{\partial x_i x_j} (t_n,x_n; \theta).
   \]
  \textcolor{black}{Using (\ref{LimIto}), $\nabla_{\theta} G_1$ is approximated as $\tilde G_1$ with a fixed constant $\Delta > 0$:}
\begin{eqnarray}
\textcolor{black}{ \tilde G_1(\theta_n, s_n) } &\textcolor{black}{\vcentcolon =}& \textcolor{black}{2} \bigg{(} \frac{\partial f}{\partial t}(t_n ,x_n; \theta_n) +  \mathcal{L}_1 f(t_n,x_n; \theta_n)  +  \frac{1}{2} \sum_{i=1}^d  \frac{ \frac{\partial f}{\partial x_i}(t, x_n + \sigma(x_n) W_{\Delta}; \theta ) - \frac{\partial f}{\partial x_i}(t,x_n; \theta)  }{\Delta}   \sigma_i(x_n)   W_{\Delta}^i \bigg{)} \notag \\
& \times & \nabla_{\theta} \bigg{(} \frac{\partial f}{\partial t} (t_n ,x_n; \theta_n) +  \mathcal{L}_1 f(t_n,x_n; \theta_n)  + \frac{1}{2} \sum_{i=1}^d  \frac{ \frac{\partial f}{\partial x_i}(t, x_n + \sigma(x_n) \tilde W_{\Delta}; \theta ) - \frac{\partial f}{\partial x_i}(t,x_n; \theta)  }{\Delta}    \sigma_i(x_n)  \tilde W_{\Delta}^i  \bigg{)}, \notag
\end{eqnarray}
where $ W_{\Delta}$ is a $d$-dimensional normal random variable with $\mathbb{E}[W_{\Delta} ] = 0$ and $\textrm{Cov}[ ( W_{\Delta} )_i, ( W_{\Delta} )_j ] = \rho_{i,j} \Delta  $.  $\tilde W_{\Delta}$ has the same distribution as $W_{\Delta}$.  $W_{\Delta}$ and $\tilde W_{\Delta}$ are independent.  $\tilde G_1(\theta_n, s_n)$ is a Monte Carlo approximation of $\nabla_{\theta} G_1(\theta_n, s_n)$.  $\tilde G_1(\theta_n, s_n)$ has $\mathcal{O}(\sqrt{\Delta})$ bias as an approximation for   $\nabla_{\theta} G_1(\theta_n, s_n)$.  This approximation error can be further improved via the following scheme using ``antithetic variates":

\begin{eqnarray}
 \tilde G_1(\theta_n, s_n)  &\textcolor{black}{\vcentcolon =}&   \tilde G_{1,a}(\theta_n, s_n)  +  \tilde G_{1,b}(\theta_n, s_n)   \label{tildeJ1OrderDelta} \\
 \tilde G_{1,a} (\theta_n, s_n) &\textcolor{black}{\vcentcolon =}&  \bigg{(} \frac{\partial f}{\partial t}(t_n ,x_n; \theta_n) +  \mathcal{L}_1 f(t_n,x_n; \theta_n)  +  \frac{1}{2} \sum_{i=1}^d \frac{ \frac{\partial f}{\partial x_i}(t, x_n + \sigma(x_n)  W_{\Delta}; \theta ) - \frac{\partial f}{\partial x_i}(t,x_n; \theta)  }{\Delta}   \sigma_i(x_n)   W_{\Delta}^i \bigg{)} \notag \\
& \times & \nabla_{\theta} \bigg{(}  \frac{\partial f}{\partial t} (t_n ,x_n; \theta_n) +  \mathcal{L}_1 f(t_n,x_n; \theta_n)  + \frac{1}{2} \sum_{i=1}^d \frac{ \frac{\partial f}{\partial x_i}(t, x_n + \sigma(x_n) \tilde W_{\Delta}; \theta ) - \frac{\partial f}{\partial x_i}(t,x_n; \theta)  }{\Delta}   \sigma_i(x_n) \tilde   W_{\Delta}^i  \bigg{)}, \notag \\
\tilde G_{1,b} (\theta_n, s_n) &\textcolor{black}{\vcentcolon =}&    \bigg{(} \frac{\partial f}{\partial t}(t_n ,x_n; \theta_n) +  \mathcal{L}_1 f(t_n,x_n; \theta_n)  - \frac{1}{2} \sum_{i=1}^d  \frac{ \frac{\partial f}{\partial x_i}(t, x_n - \sigma(x_n)  W_{\Delta}; \theta ) - \frac{\partial f}{\partial x_i}(t,x_n; \theta)  }{\Delta}   \sigma_i(x_n)  W_{\Delta}^i \bigg{)} \notag \\
& \times & \nabla_{\theta}  \bigg{(} \frac{\partial f}{\partial t} (t_n ,x_n; \theta_n) +   \mathcal{L}_1 f(t_n,x_n; \theta_n)  - \frac{1}{2}  \sum_{i=1}^d \frac{ \frac{\partial f}{\partial x_i}(t, x_n - \sigma(x_n) \tilde W_{\Delta}; \theta ) - \frac{\partial f}{\partial x_i}(t,x_n; \theta)  }{\Delta}   \sigma_i(x_n)  \tilde  W_{\Delta}^i  \bigg{)}. \notag  
\end{eqnarray}
\textcolor{black}{The approximation (\ref{tildeJ1OrderDelta}) has $\mathcal{O}( \Delta )$ bias as an approximation for $ \nabla_{\theta} G_1(\theta_n, s_n)$. (\ref{tildeJ1OrderDelta})  uses antithetic variates in the sense that $ \tilde G_{1,a}(\theta_n, s_n)$ uses the random variables $(W_{\Delta}, \tilde W_{\Delta})$ while $\tilde G_{1,b} (\theta_n, s_n) $ uses $(-W_{\Delta}, -\tilde W_{\Delta})$. See \cite{Glynn} for a background on antithetic variates in simulation algorithms. A Taylor expansion can be used to show the approximation error is $\mathcal{O}( \Delta )$.} It is important to highlight that there is no computational cost associated with the magnitude of $\Delta$; an arbitrarily small $\Delta$ can be chosen with no additional computational cost (although there may be numerical underflow or overflow problems). The modified algorithm using the Monte Carlo approximation for the second derivatives is:

\begin{enumerate}
\item Generate random points $(t_n, x_n)$ from $[0,T] \times \Omega$ and $(\tau_n, z_n)$ from $[0,T] \times \partial \Omega$ according to respective densities $\nu_1$ and $\nu_2$.  Also, draw the random point $w_n$ from $\Omega$ with density $\nu_3$.
\item \textcolor{black}{Calculate the step $ \tilde G(\theta_n ,s_n) =  \tilde G_1(\theta_n , s_n )  + \nabla_{\theta} G_2(\theta_n , s_n )  + \nabla_{\theta} G_3(\theta_n , s_n )$ at the randomly sampled points $s_n = \{ (t_n, x_n), ( \tau_n, z_n), w_n \}$. $\tilde G(\theta_n ,s_n)$ is an approximation for $\nabla_{\theta} G(\theta_n ,s_n)$.}
\item Take a step at the random point $s_n$:
\begin{eqnarray}
\theta_{n+1} = \theta_n - \alpha_n \tilde G(\theta_n, s_n)  \notag
\end{eqnarray}
\item Repeat until convergence criterion is satisfied.
\end{enumerate}

In conclusion, the modified algorithm here is computationally less expensive than the original algorithm in Section \ref{Algorithm} but introduces some bias and variance. \textcolor{black}{The variance essentially increases the i.i.d. noise in the stochastic gradient descent step; this noise averages out over a large number of samples though.} The original algorithm in Section \ref{Algorithm} is unbiased and has lower variance, but is computationally more expensive. We numerically implement the algorithm for a class of free boundary PDEs in Section \ref{NumericalAnalysis}. \textcolor{black}{Future research may investigate other methods to further improve the computational evaluation of the second derivative terms (for instance, multi-level Monte Carlo).}


\section{Numerical Analysis for a High-dimensional Free Boundary PDE} \label{NumericalAnalysis}
We test our algorithm on a class of high-dimensional free boundary PDEs. These free boundary PDEs are used in finance to price American options and are often referred to as ``American option PDEs". An American option is a financial derivative on a portfolio of stocks. The option owner may at any time $t \in [0,T]$ choose to exercise the American option and receive a payoff which is determined by the underlying prices of the stocks in the portfolio. $T$ is called the maturity date of the option and the payoff function is $g(x): \mathbb{R}^d \rightarrow \mathbb{R}$. Let $X_t \in \mathbb{R}^d$ be the prices of $d$ stocks.
If at time $t$ the stock prices $X_t = x$, the price of the option is $u(t,x)$.  The price function $u(t,x)$ satisfies a free boundary PDE on $[0,T] \times \mathbb{R}^d$. For American options, one is primarily interested in the solution $u(0, X_0)$ since this is the fair price to buy or sell the option.

Besides the high dimensions and the free boundary, the American option PDE is challenging to numerically solve since the payoff function $g(x)$ (which both appears in the initial condition and determines the free boundary) is not continuously differentiable.

Section \ref{StatementofPDE} states the free boundary PDE and the deep learning algorithm to solve it. \textcolor{black}{To address the free boundary, we supplement the algorithm presented in Section \ref{Algorithm} with an iterative method;  see Section \ref{StatementofPDE}.} Section \ref{Implementation} describes the architecture and implementation details for the neural network. Section \ref{SemiAnalytic} reports numerical accuracy for a case where a semi-analytic solution exists.  Section \ref{NoSemiAnalytic} reports numerical accuracy for a case where no semi-analytic solution exists.


\subsection{The Free Boundary PDE} \label{StatementofPDE}
We now specify the free boundary PDE for $u(t,x)$. The stock price dynamics and option price are:
\begin{eqnarray*}
d X_t^i &=&  \mu(X_t^i) dt +\sigma(X_t^i) d W_t^i, \notag \\
u(t,x) &=&  \sup_{\tau \geq t} \mathbb{E}[ e^{- r( \tau \wedge T)} g(X_{\tau \wedge T}) | X_t = x],
\label{SDEx}
\end{eqnarray*}
where $W_t \in \mathbb{R}^d$ is a standard Brownian motion and $\textrm{Cov}[ dW_t^i, dW_t^j] = \rho_{i,j} dt$.  The price of the American option is $u(0, X_0)$.

\textcolor{black}{The model (\ref{SDEx}) for the stock price dynamics is widely used in practice and captures several desirable characteristics. First, the drift $\mu(x)$ measures the ``average" growth in the stock prices. The Brownian motion $W_t$ represents the randomness in the stock price, and the magnitude of the randomness is given by the coefficient function $\sigma(X_t^i)$. The movement of stock prices are correlated (e.g., if Microsoft's price increases, it is likely that Apple's price will also increase). The magnitude of the correlation between two stocks $i$ and $j$ is specified by the parameter $\rho_{i,j}$. An example is the well-known Black-Scholes model $\mu(x) = \mu x$ and $\sigma(x) = \sigma x$. In the Black-Scholes model, the average rate of return for each stock is $\mu$.}

\textcolor{black}{An American option is a financial derivative which the owner can choose to ``exercise" at any time $t \in [0,T]$. If the owner exercises the option, they receive the financial payoff $g(X_t)$ where $X_t$ is the prices of the underlying stocks. If the owner does not choose to exercise the option, they receive the payoff $g(X_T)$ at the final time $T$. The value (or price) of the American option at time $t$ is $u(t,X_t)$.} \textcolor{black}{Some typical examples of the payoff function $g(x): \mathbb{R}^d \rightarrow \mathbb{R}$ are $g(x)=\max \big{(} (\prod_{i=1}^d x_i )^{1/d}- K, 0 \big{)}$ and $ g(x)=\max \big{(} \frac{1}{d} \sum_{i=1}^d x_i - K, 0 \big{)}$. The former is referred to as a ``geometric payoff function" while the latter is called an ``arithmetic payoff function."} \textcolor{black}{$K$ is the ``strike price" and is a positive number.}

The price function $u(t,x)$ in (\ref{SDEx}) is the solution to a free boundary PDE and will satisfy:
\begin{eqnarray}
0 &=& \frac{ \partial u}{\partial t}(t,x)  + \mu(x) \cdot \frac{ \partial u}{\partial x}(t,x)  + \frac{1}{2} \sum_{i,j =1}^d \rho_{i,j} \sigma(x_i) \sigma(x_j) \frac{\partial^2 u}{\partial x_i \partial x_j }(t,x) -r u(t,x),  \phantom{....}  \forall \phantom{..} \big{\{} (t,x) : u(t,x) > g(x) \big{\}}.  \notag \\
u(t,x) & \geq &  g(x),  \phantom{....}  \forall \phantom{..} (t,x).  \notag \\
u(t,x) &\in& C^1(\mathbb{R}_{+} \times \mathbb{R}^d), \phantom{....}  \forall \phantom{..} \big{\{} (t,x) : u(t,x) = g(x) \big{\}}. \notag \\
u(T,x) &=&  g(x),  \phantom{....}  \forall \phantom{..} x.
\label{AmericanOptionPDE}
\end{eqnarray}
The free boundary set is $F = \big{\{} (t,x) : u(t,x) = g(x) \big{\}}$.  $u(t,x)$ satisfies a partial differential equation ``above" the free boundary set $F$, and $u(t,x)$ equals the function $g(x)$ ``below" the free boundary set $F$.

The deep learning algorithm for solving the PDE (\ref{AmericanOptionPDE}) requires simulating points above and below the free boundary set $F$.  We use an iterative method to address the free boundary. The free boundary set $F$ is approximated using the current parameter estimate $\theta_{n}$. This approximate free boundary is used in the probability measure that we simulate points with. The gradient is not taken with respect to the $\theta_n$ input of the probability density used to simulate random points.  For this purpose, define the objective function:
\begin{eqnarray*}
J(f;\theta, \tilde \theta)  &=&   \norm{ \frac{ \partial f}{\partial t}(t,x; \theta)  + \mu(x) \cdot \frac{ \partial f}{\partial x}(t,x; \theta)  + \frac{1}{2} \sum_{i,j =1}^d \rho_{i,j} \sigma(x_i) \sigma(x_j) \frac{\partial^2 f}{\partial x_i \partial x_j }(t,x; \theta) -r f(t,x; \theta) }_{[0,T] \times \Omega, \nu_1(\tilde \theta) }^2 \notag \\
&+&   \norm{  \max ( g(x) - f(t,x; \theta), 0 ) }_{ [0,T] \times \Omega, \nu_2(\tilde \theta) }^2  \notag \\
&+& \norm{  f(T,x; \theta ) - g(x) }_{ \Omega, \nu_3 }^2.
\label{ObjectiveAmericanOptionPDE}
\end{eqnarray*}
Descent steps are taken in the direction $-\nabla_{\theta} J( f;\theta, \tilde \theta)$.  $\nu_1(\tilde \theta)$ and $\nu_2(\tilde \theta)$ are the densities of the points in $\tilde B^1$ and $\tilde B^2$, which are defined below.  The deep learning algorithm is:

\begin{enumerate} \label{BatchDGM}
\item Generate the random batch of points $B^1 = \{ t_m, x_m \}_{m=1}^M$ from $[0,T] \times \Omega$ according to the probability density $\nu_1^0$.  Select the points $\tilde B^1 = \{ (t,x) \in B^1 : f(t,x; \theta_n) > g(x) \}.$
\item Generate the random batch of points $B^2 = \{ \tau_m, z_m \}_{m=1}^M$ from $[0,T] \times \partial \Omega$ according to the probability density $\nu_2^0$.  Select the points $\tilde B^2 = \{ (\tau,z) \in B^2 : f(\tau, z ;  \theta_n) \leq g(z) \}.$
\item Generate the random batch of points $B^3 = \{ w_m \}_{m=1}^M$ from $\Omega$ with probability density $\nu_3$.
\item Approximate $J(f;\theta_n, \tilde \theta_n)$ as $J(f;\theta_n, S_n)$ at the randomly sampled points $S_n = \{ \tilde B^1, \tilde B^2, B^3 \}$:
\begin{eqnarray}
J(f;\theta_n , S_n) &=&  \frac{1}{|\tilde B^1 | } \sum_{(t_m, x_m) \in \tilde B^1 } \bigg{(} \frac{ \partial f}{\partial t}(t_m,x_m; \theta_n)  + \mu(x_m) \cdot \frac{ \partial f}{\partial x}(t_m,x_m; \theta_n)   \notag \\
&+& \frac{1}{2} \sum_{i,j =1}^d \rho_{i,j} \sigma(x_i) \sigma(x_j) \frac{\partial^2 f}{\partial x_i \partial x_j }(t_m,x_m; \theta_n) -r f(t_m,x_m; \theta_n) \bigg{)}^2  \notag \\
&+&    \frac{1}{ | \tilde B^2 | } \sum_{(\tau_m, z_m) \in \tilde B^2 } \max \big{(} g(z_m) - f(\tau_m,z_m; \theta_n), 0 \big{)}^2 \notag \\
&+&     \frac{1}{| B^3 |} \sum_{ w_m \in B^3}  \bigg{(} f(T,w_m; \theta ) - g(w_m)  \bigg{)}^2. \notag
\end{eqnarray}
\item Take a descent step for the random batch $S_n$:
\begin{eqnarray}
\theta_{n+1} = \theta_n - \alpha_n \nabla_{\theta} J(f;\theta_n, S_n). \notag
\end{eqnarray}
\item Repeat until convergence criterion is satisfied.

\end{enumerate}
\textcolor{black}{The second derivatives in the above algorithm can be approximated using the method from Section \ref{ModifiedAlgorithm}.}

\subsection{\textcolor{black}{Implementation details for the algorithm}} \label{Implementation}
\textcolor{black}{This section provides details for the implementation of the algorithm, including the DGM network architecture, hyperparameters, and computational approach.}

\textcolor{black}{The architecture of a neural network can be crucial to its success. Frequently, different applications require different architectures. For example, convolution networks are essential for image recognition while long short-term networks (LSTMs) are useful for modeling sequential data. Clever choices of architectures, which exploit a priori knowledge about an application, can significantly improve performance. In the PDE applications in this paper, we found that a neural network architecture similar in spirit to that of LSTM networks improved performance.}

\textcolor{black}{The PDE solution requires a model $f(t, x; \theta)$ which can make ``sharp turns" due to the final condition, which is of the form $u(T,x) = \max( p(x), 0)$ (the first derivative is discontinuous when $p(x) = 0$). The shape of the solution $u(t,x)$ for $t < T$, although ``smoothed" by the diffusion term in the PDE, will still have a nonlinear profile which is rapidly changing in certain spatial regions. In particular, we found the following network architecture to be effective:}

\textcolor{black}{\begin{eqnarray}
S^{1} &=&  \sigma( W^1 \overset{\rightarrow} x + b^1 ), \notag \\
Z^{\ell} &=&  \sigma ( U^{z,\ell} \overset{\rightarrow} x  + W^{z,\ell} S^{\ell} + b^{z,\ell} ), \phantom{....} \ell = 1, \ldots, L, \notag \\
G^{\ell} &=&  \sigma ( U^{g,\ell} \overset{\rightarrow} x  + W^{g,\ell} S^1 + b^{g,\ell} ), \phantom{....} \ell = 1, \ldots, L, \notag \\
R^{\ell} &=&  \sigma ( U^{r,\ell} \overset{\rightarrow} x  + W^{r,\ell} S^{\ell} + b^{r, \ell} ), \phantom{....} \ell = 1, \ldots, L, \notag \\
H^{\ell} &=&  \sigma ( U^{h,\ell} \overset{\rightarrow} x  + W^{h,\ell}  ( S^{\ell} \odot R^{\ell} ) + b^{h, \ell} ), \phantom{....} \ell = 1, \ldots, L, \notag \\
S^{\ell+1} &=&  (1 - G^{\ell} ) \odot H^{\ell} + Z^{\ell} \odot S^{\ell}, \phantom{....} \ell = 1, \ldots, L, \notag \\
f(t,x; \theta) &=& W S^{L+1} + b,
\label{ArchitectureNeuralNetwork}
\end{eqnarray}
where $\overset{\rightarrow} x  = (t,x)$, \textcolor{black}{the number of hidden layers is $L+1$}, and $\odot$ denotes element-wise multiplication (i.e., $z \odot v = \big{(} z_0 v_0, \ldots, z_N v_N \big{)}$). The parameters are
\begin{eqnarray}
\theta = \bigg{\{} W^1, b^1, \bigg{(} U^{z, \ell}, W^{z, \ell}, b^{z, \ell} \bigg{)}_{\ell=1}^L,  \bigg{(} U^{g, \ell}, W^{g, \ell}, b^{g, \ell} \bigg{)}_{\ell=1}^L, \bigg{(} U^{r, \ell}, W^{r, \ell}, b^{r, \ell} \bigg{)}_{\ell=1}^L, \bigg{(} U^{h, \ell}, W^{h, \ell}, b^{h, \ell} \bigg{)}_{\ell=1}^L, W, b \bigg{\}}. \notag
\end{eqnarray}
The number of units in each layer is $M$ and $\sigma: \mathbb{R}^M \rightarrow \mathbb{R}^M$ is an element-wise nonlinearity:
\begin{eqnarray}
\sigma(z) = \bigg{(} \phi(z_1), \phi(z_2), \ldots, \phi(z_M) \bigg{)},
\end{eqnarray}
where $\phi: \mathbb{R} \rightarrow \mathbb{R}$ is a nonlinear activation function such as the tanh function, sigmoidal function $\frac{e^y}{1 + e^y}$, or rectified linear unit (ReLU) $\max(y,0)$. The parameters in $\theta$ have dimensions $W^1 \in \mathbb{R}^{M \times (d+1)}$, $b^1 \in \mathbb{R}^M$, $U^{z,\ell} \in \mathbb{R}^{M \times (d +1)}$, $W^{z, \ell} \in \mathbb{R}^{M \times M}$, $b^{z, \ell} \in \mathbb{R}^M$, $U^{g,\ell} \in \mathbb{R}^{M \times (d +1)}$, $W^{g, \ell} \in \mathbb{R}^{M \times M}$, $b^{g, \ell} \in \mathbb{R}^M$, $U^{r,\ell} \in \mathbb{R}^{M \times (d +1)}$, $W^{r, \ell} \in \mathbb{R}^{M \times M}$, $b^{r, \ell} \in \mathbb{R}^M$, $U^{h, \ell} \in \mathbb{R}^{M \times (d +1)}$, $W^{h, \ell} \in \mathbb{R}^{M \times M}$, $b^{h, \ell} \in \mathbb{R}^M$, $W \in \mathbb{R}^{1 \times M}$, and $b \in \mathbb{R}$.}

\textcolor{black}{The architecture (\ref{ArchitectureNeuralNetwork}) is relatively complicated. Within each layer, there are actually many ``sub-layers" of computations. The important feature is the repeated element-wise multiplication of nonlinear functions of the input. This helps to model more complicated functions which are rapidly changing in certain time and space regions.  The neural network architecture (\ref{ArchitectureNeuralNetwork}) is similar to the architecture for LSTM networks (see \cite{SchmidhuberLSTM}) and highway networks (see \cite{SchmidhuberHighway}). }

\textcolor{black}{The key hyperparameters in the neural network (\ref{ArchitectureNeuralNetwork}) are the number of layers $L$, the number of units $M$ in each sub-layer, and the choice of activation unit $\phi(y)$. We found for the applications in this paper that the hyperparameters $L = 3$ \textcolor{black}{(i.e., four hidden layers)}, $M = 50$, and $\phi(y) = \tanh(y)$ were effective. It is worthwhile to note that the choice of $\phi(y) = \tanh(y)$ means that $f(t,x; \theta)$ is smooth and therefore can solve for a ``classical solution" of the PDE. The neural network parameters are initialized using the Xavier initialization (see \cite{Xavier}). The architecture (\ref{ArchitectureNeuralNetwork}) is bounded in the input $x$ (for a fixed choice of parameters $\theta$) if $\sigma(\cdot)$ is a tanh or sigmoidal function; it may be helpful to allow the network to be unbounded for approximating unbounded/growing functions. We found that replacing the $\sigma(\cdot)$ in the $H^{\ell}$ sub-layer with the identity function can be an effective way to develop an unbounded network.}

\textcolor{black}{We emphasize that the only input to the network is $(t,x)$. We do not use any custom-designed nonlinear transformations of $(t,x)$. If properly chosen, such additional inputs might help performance. For example, the European option PDE solution (which has an analytic formula) could be included as an input.}

\textcolor{black}{A regularization term (such as an $\ell^2$ penalty) could also be included in the objective function for the algorithm. Such regularization terms are used for reducing overfitting in machine learning models estimated using datasets which have a limited number of data samples. (For example, a model estimated on a dataset of $60,000$ images.) However, it's unclear if this will be helpful in the context of this paper's application, since there is no strict upper bound on the size of the dataset (i.e., one can always simulate more time/space points).}

\textcolor{black}{Our computational approach to training the neural network involved several components. The second derivatives are approximated using the method from Section \ref{ModifiedAlgorithm}. Training is distributed across $6$ GPU nodes using asynchronous stochastic gradient descent (we provide more details on this below). Parameters are updated using the well-known ADAM algorithm (see \cite{ADAM}) with a decaying learning rate schedule (more details on the learning rate are provided below).  Accuracy can be improved by calculating a running average of the neural network solutions over a sequence of training iterations (essentially a computationally cheap approach for building a model ensemble). We also found that model ensembles (of even small sizes of 5) can slightly increase accuracy.}

\textcolor{black}{Training of the neural network is distributed across several GPU nodes in order to accelerate training. We use asynchronous stochastic gradient descent, which is a widely-used method for parallelizing training of machine learning models. On each node, i.i.d. space and time samples are generated. Each node calculates the gradient of the objective function with respect to the parameters on its respective batch of simulated data. These gradients are then used to update the model, which is stored on a central node called a ``parameter server". Figure \ref{ASGDjcp} displays the computational setup. Updates occur asynchronously; that is, node $i$ updates the model immediately upon completion of its work, and does not wait for node $j$ to finish its work. The ``work" here is calculating the gradients for a batch of simulated data. Before a node calculates the gradient for a new batch of simulated data, it receives an updated model from the parameter server. For more details on asynchronous stochastic gradient descent, see \cite{Dean}.}

\begin{figure}[h!]
\begin{center}
\includegraphics[width=.8\textwidth, height=40mm]{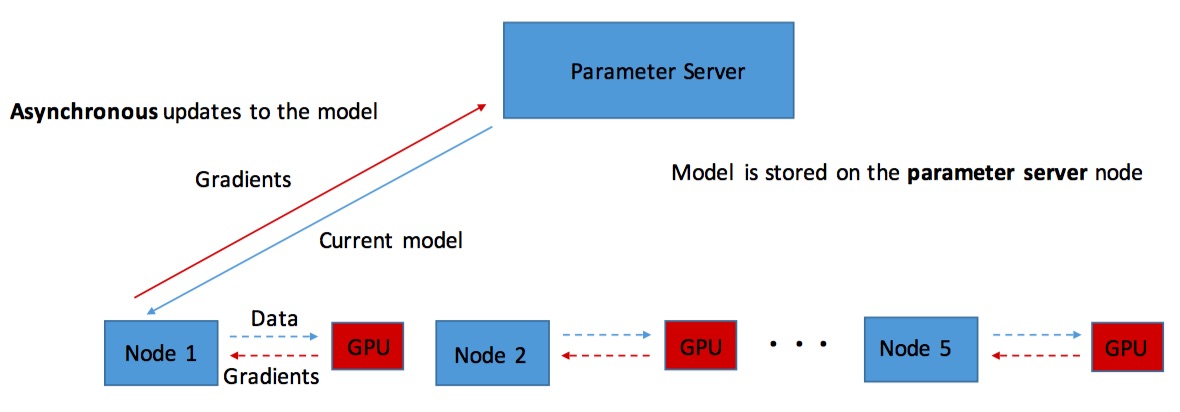}
\end{center}
\caption{Asynchronous stochastic gradient descent on a cluster of GPU nodes.}
\label{ASGDjcp}
\end{figure}

\textcolor{black}{During training, we decrease the learning as the number of iterations increases. We use a learning rate schedule where the learning rate is a piecewise constant function of the number of iterations. This is a typical choice. We found the following learning rate schedule to be effective:}

 \textcolor{black}{\[   \alpha_n = \left\{
\begin{array}{ll}
      10^{-4} & n \leq 5,000 \\
       5 \times 10^{-4}  &   5,000 < n \leq 10,000 \\
       10^{-5} &  <  10,000 < n \leq 20,000 \\
      5 \times 10^{-6} & 20,000 < n \leq 30,000  \\
      10^{-6}  & 30,000 < n \leq 40,000  \\
      5 \times 10^{-7}  & 40,000 < n \leq 45,000  \\
      10^{-7}  & 45,000 <  n \\
\end{array}
\right. \]}

\textcolor{black}{We use approximately $100,000$ iterations. An ``iteration" involves batches of size $1,000$ on each of the GPU nodes. Therefore, there are $5,000$ simulated time/space points for each iteration. In total, we used approximately $500$ million simulated time/space points to train the neural network.}

\textcolor{black}{We implement the algorithm using TensorFlow and PyTorch, which are software libraries for deep learning. TensorFlow has reverse mode automatic differentiation which allows the calculation of derivatives for a broad range of functions. For example, TensorFlow can be used to calculate the gradient of the neural network (\ref{ArchitectureNeuralNetwork}) with respect to $x$ or $\theta$. TensorFlow also allows for the training of models on graphics processing units (GPUs). A GPU, which has thousands of cores, can be use to highly parallelize the training of deep learning models. We furthermore distribute our computations across multiple GPU nodes, as described above. The computations in this paper were performed on the Blue Waters supercomputer which has a large number of GPU nodes. }


\subsection{A High-dimensional Free Boundary PDE  with a Semi-Analytic Solution } \label{SemiAnalytic}

We implement our deep learning algorithm to solve the PDE (\ref{AmericanOptionPDE}).  The accuracy of our deep learning algorithm is evaluated in up to $200$ dimensions. The results are reported below in Table \ref{AmericanOptionTable1}.

\begin{table}[ht!]
\begin{center}
 \begin{tabular}{  | c  | c | }
   \hline
Number of dimensions  &   Error  \\ \hline \hline
 3               &  0.05\%    \\ \hline
 20                  &   0.03\%    \\ \hline
 100                &   0.11\%    \\ \hline
  200                &   0.22\%    \\ \hline
 \end{tabular}
\end{center}
\caption{\label{AmericanOptionTable1} The deep learning algorithm solution is compared with a semi-analytic solution for the Black-Scholes model. The parameters $\mu(x) = (r -c) x$ and $\sigma(x) = \sigma x$.  All stocks are identical with correlation $\rho_{i,j} = .75$, volatility $\sigma = .25$, initial stock price $X_0 = 1$, dividend rate $c = 0.02$, and interest rate $r = 0$.  The maturity of the option is $T =2$ and the strike price is $K=1$.  The payoff function is $ g(x)=\max \big{(} (\prod_{i=1}^d x_i )^{1/d}- K, 0 \big{)}$.  The error is reported for the price $u(0, X_0)$ of the at-the-money American call option.  The error is $\frac{ | f(0, X_0; \theta) - u(0, X_0) |}{|u(0, X_0)| } \times 100\%$.}
\end{table}

\textcolor{black}{The semi-analytic solution used in Table \ref{AmericanOptionTable1} is provided below. Let $\mu(x) = (r -c) x$, $\sigma(x) = \sigma x$, and $\rho_{i,j} = \rho$ for $i \neq j$ (i.e., the Black-Scholes model). If the payoff function in (\ref{AmericanOptionPDE}) is $ g(x)=\max \big{(} (\prod_{i=1}^d x_i )^{1/d}- K, 0 \big{)}$, then there is a semi-analytic solution to (\ref{AmericanOptionPDE}):
\begin{eqnarray}
u(t,x) &=& v(t, (\prod_{i=1}^d x_i )^{1/d}- K ),
\label{SemiAnalyticFormula}
\end{eqnarray}
where $v(t,x)$ satisfies the one-dimensional free boundary PDE
\begin{eqnarray}
0 &=& \frac{ \partial v}{\partial t}(t,x)  + \hat \mu x \frac{ \partial v}{\partial x}(t,x)  + \frac{1}{2} \hat \sigma^2 x \frac{\partial^2 v}{\partial x^2}(t,x) -r v(t,x),  \phantom{....}  \forall \phantom{..} \big{\{} (t,x) : v(t,x) > \hat g(x) \big{\}}.  \notag \\
v(t,x) & \geq &  \hat g(x),  \phantom{....}  \forall \phantom{..} (t,x).  \notag \\
\hat v(t,x) &\in& C^1(\mathbb{R}_{+} \times \mathbb{R}^d), \phantom{....}  \forall \phantom{..} \big{\{} (t,x) : v(t,x) = \hat g(x) \big{\}}. \notag \\
v(T,x) &=&  \hat g(x),  \phantom{....}  \forall \phantom{..} x,
\label{AmericanOptionPDEvOneD}
\end{eqnarray}
where $\hat \sigma^2 = \frac{ d \sigma^2 + d(d-1) \rho \sigma^2 }{d^2}$, $\hat \mu = (r-c) - \frac{1}{2} \hat \sigma^2 + \frac{1}{2} \sigma^2$, and $\hat g(x) = \max(x, 0)$. The one-dimensional PDE (\ref{AmericanOptionPDEvOneD}) can be solved using finite difference methods. If $f(t,x; \theta)$ is the deep learning algorithm's estimate for the PDE solution at $(t,x)$, the relative error at the point
$(t,x)$ is $\frac{ | f(t, x; \theta) - u(t, x) |}{|u(t,x) | } \times 100\%$ and the absolute error at the point $(t,x)$ is $| f(t, x; \theta) - u(t, x) |$. The relative error and absolute error at the point $(t,x)$ can be evaluated using the semi-analytic solution (\ref{SemiAnalyticFormula}).}


Although the solution at $(t,x) = (0, X_0)$ is of primary interest for American options, most other PDE applications are interested in the entire solution $u(t,x)$. The deep learning algorithm provides an approximate solution across all time and space $(t,x) \in [0, T] \times \Omega$. As an example, we present in Figure \ref{SubplotFigureContour} contour plots of the absolute error and percent error across time and space for the American option PDE in $20$ dimensions. The contour plot is produced in the following way:
\begin{enumerate}
\item Sample time points $t^\ell$ uniformly on $[0,T]$ and sample spatial points $x^\ell = (x_1^\ell, \ldots, x_{20}^\ell)$ from the joint distribution of $X_t^1, \ldots, X_t^{20}$ in equation (\ref{SDEx}). This produces an ``envelope" of sampled points since $X_t$ spreads out as a diffusive process from $X_0 = 1$.
\item  Calculate the error $E^\ell$ at each sampled point $(t^\ell, x^\ell)$ for $\ell =1, \ldots, L$.
\item Aggregate the error over a two-dimensional subspace $\big{(} t^\ell, \displaystyle (\prod_{i=1}^{20} x_i^\ell )^{1/20}, E^\ell \big{)}$ for $\ell =1, \ldots, L$.
\item Produce a contour plot from the data $\big{(}  t^\ell, (\prod_{i=1}^{20} x_i^\ell )^{1/20}, E^\ell \big{)}_{\ell=1}^L$. The x-axis is $t$ and the y-axis is the geometric average $(\prod_{i=1}^{20} x_i )^{1/20}$, which corresponds to the final condition $g(x)$.
\end{enumerate}
Figure \ref{SubplotFigureContour} reports both the absolute error and the percent error. The percent error $\frac{ | f(t,x; \theta) - u(t,x) |}{ | u(t,x) | } \times 100\%$ is reported for points where $|u(t,x)| > 0.05$. The absolute error becomes relatively large in a few areas; however, the solution $u(t,x)$ also grows large in these areas and therefore the percent error remains small.

\begin{figure}[h!]
\begin{center}
\includegraphics[width=.7\textwidth, height=80mm]{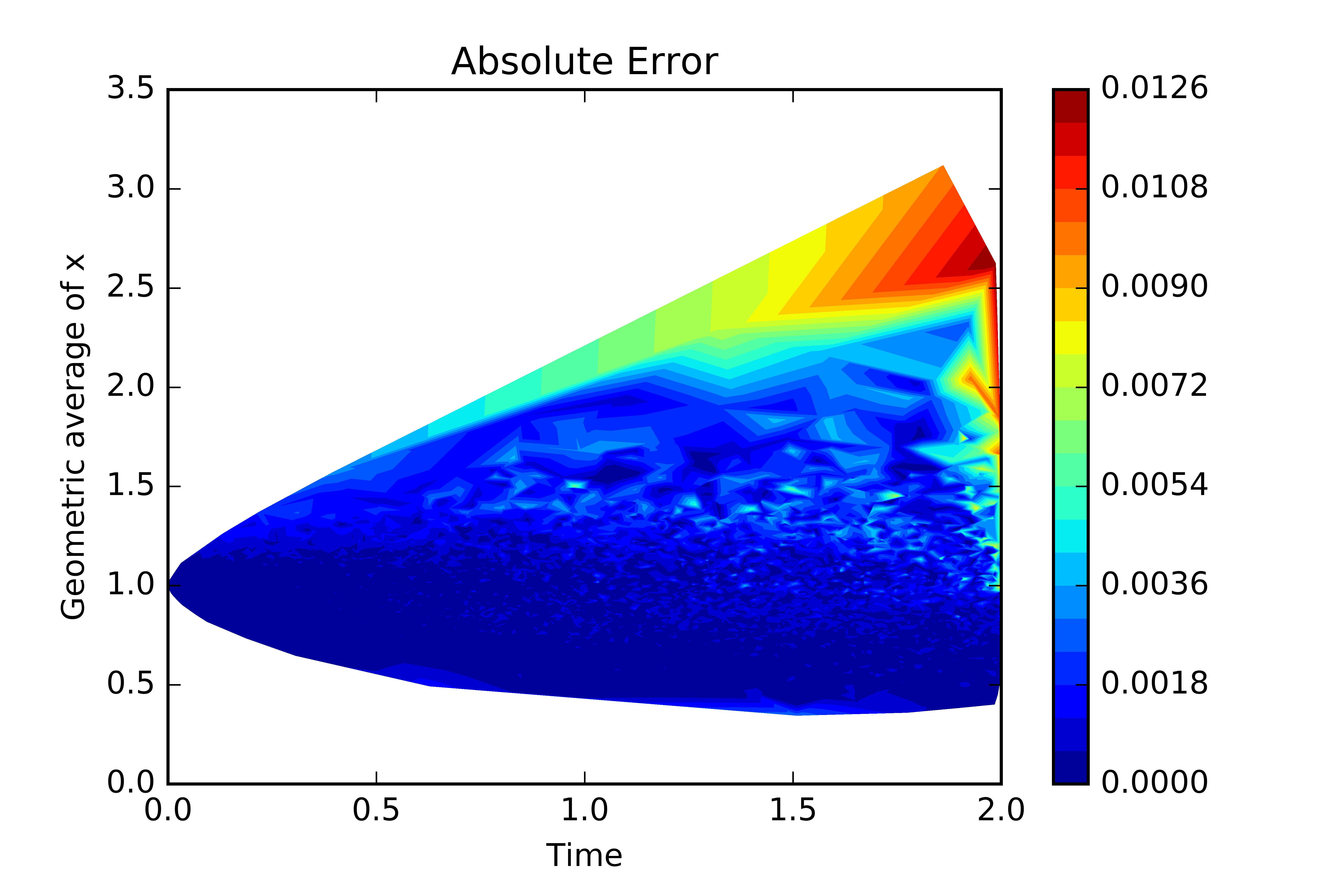}
\includegraphics[width=.7\textwidth, height=80mm]{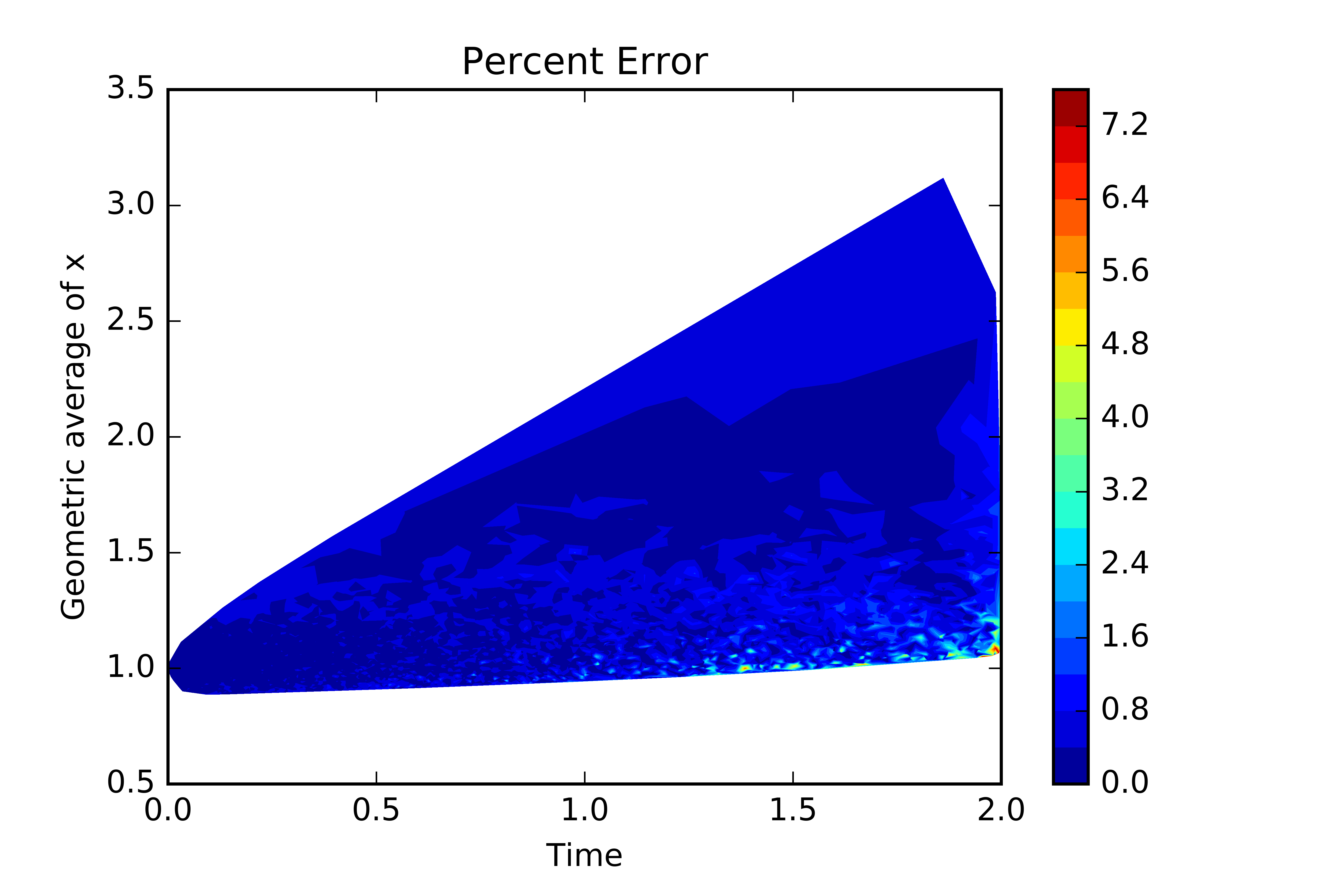}
\end{center}
\caption{Top: Absolute error. Bottom: Percent error. For reference, the price at time $0$ is $0.1003$ and the solution at time $T$ is $\max(\textrm{geometric average of x} - 1, 0 )$.}
\label{SubplotFigureContour}
\end{figure}

%

\subsection{A High-dimensional Free Boundary PDE \emph{without} a Semi-Analytic Solution } \label{NoSemiAnalytic}

We now consider a case of the American option PDE which does not have a semi-analytic solution. The American option PDE has the special property that it is possible to calculate error bounds on an approximate solution. Therefore, we can evaluate the accuracy of the deep learning algorithm even on cases where no semi-analytic solution is available.

We previously only considered a symmetrical case where $\rho_{i,j} = 0.75$ and $\sigma = 0.25$ for all stocks. This section solves a more challenging heterogeneous case where $\rho_{i,j}$ and $\sigma_i$ vary across all dimensions $i =1, 2, \ldots, d$.  The coefficients are fitted to actual data for the stocks IBM, Amazon, Tiffany, Amgen, Bank of America, General Mills, Cisco, Coca-Cola, Comcast, Deere, General Electric, Home Depot, Johnson \& Johnson, Morgan Stanley, Microsoft, Nordstrom, Pfizer, Qualcomm, Starbucks, and Tyson Foods from 2000-2017. This produces a PDE with widely-varying coefficients for each of the $\frac{d^2 +d}{2}$ second derivative terms.  The correlation coefficients $\rho_{i,j}$ range from $-0.53$ to $0.80$ for $i \neq j$ and $\sigma_i$ ranges from $0.09$ to $0.69$.


Let $f(t,x; \theta)$ be the neural network approximation.  \cite{ChrisRogers} derived that the PDE solution $u(t,x)$ lies in the interval:
\begin{eqnarray}
u(t,x) &\in&  \big[ \underline{u}(t,x), \overline{u}(t,x) \big], \notag \\
\underline{u}(t,x) &=& \mathbb{E}  \bigg [ g( X_{\tau} ) | X_t = x, \tau > t  \bigg], \notag \\
\overline{u}(t,x) &=& \mathbb{E} \bigg [  \sup_{s \in [t,T] }  \big[ e^{-r (s-t) } g( X_s)  - M_s \big ] \bigg ].
\label{AmericanOptionBounds}
\end{eqnarray}
where $\tau = \inf \{ t \in [0,T]: f(t, X_t; \theta) < g(X_t ) \}$ and $M_s$ is a martingale constructed from the approximate solution $f(t,x; \theta)$
\begin{eqnarray}
M_s &=& f(s,X_s; \theta) - f(t,X_t; \theta)  - \int_t^s \bigg{[}\frac{ \partial f}{\partial t}(s',X_{s'}; \theta)  + \mu(X_{s'}) \frac{ \partial f}{\partial x}(s',X_{s'}; \theta)  \notag \\
&+& \frac{1}{2} \sum_{i,j =1}^d \sigma(X_{s', i}) \sigma(X_{s', j}) \frac{\partial^2 f}{\partial x_i \partial x_j }(s',X_{s'}; \theta) -r f(s',X_{s'}; \theta) \bigg{]} ds'. 
\label{integralM}
\end{eqnarray}
\textcolor{black}{The bounds (\ref{AmericanOptionBounds})} depend only on the approximation $f(t,x; \theta)$, which is known, and can be evaluated via Monte Carlo simulation. The integral for $M_s$ must also be discretized. The best estimate for the price of the American option is the midpoint of the interval $[\underline{u}(0, X_0), \overline{u}(0,X_0) ]$, which has an error bound of $\frac{ \overline{u}(0,X_0) - \underline{u}(0, X_0)  }{  2  \underline{u}(0, X_0) } \times 100\%$. Numerical results are in Table \ref{AmericanOptionTable2}.

\begin{table}[ht!]
\begin{center}
 \begin{tabular}{  | c  | c | c| c| c| }
   \hline
Strike price &  Neural network solution & Lower Bound & Upper Bound & Error bound \\ \hline \hline
0.90     &  0.14833   & 0.14838    & 0.14905  &  0.23\%       \\ \hline
0.95      &  0.12286  &  0.12270   & 0.12351 &  0.33\%      \\ \hline
1.00      &  0.10136  &   0.10119  & 0.10193  &   0.37\%      \\ \hline
1.05      &  0.08334   & 0.08315   &  0.08389 & 0.44\%       \\ \hline
1.10      & 0.06841    &  0.06809   &  0.06893 &  0.62\%       \\ \hline
 \end{tabular}
\end{center}
\caption{\label{AmericanOptionTable2}  The accuracy of the deep learning algorithm is evaluated on a case where there is no semi-analytic solution.  The parameters $\mu(x) = (r -c) x$ and $\sigma(x) = \sigma x$.  The correlations $\rho_{i,j}$ and volatilities $\sigma_i$ are estimated from data to generate a heterogeneous diffusion matrix. The initial stock price is $X_0 = 1$, dividend rate $c = 0.02$, and interest rate $r = 0$ for all stocks. The maturity of the option is $T =2$. The payoff function is $ g(x)=\max \big{(} \frac{1}{d} \sum_{i=1}^d x_i- K, 0 \big{)}$.  The neural network solution and its error bounds are reported for the price $u(0, X_0)$ of the American call option.  The best estimate for the price of the American option is the midpoint of the interval $[\underline{u}(0, X_0), \overline{u}(0,X_0) ]$, which has an error bound of $\frac{ \overline{u}(0,X_0) - \underline{u}(0, X_0)  }{  2  \underline{u}(0, X_0) } \times 100\%$. In order to calculate the upper bound, the integral (\ref{integralM}) is discretized with time step size $\Delta = 5 \times 10^{-4}$.}
\end{table}

We present in Figure \ref{SubplotFigureContour2} contour plots of the absolute error bound and percent error bound across time and space for the American option PDE in $20$ dimensions for strike price $K = 1$. The contour plot is produced in the following way:
\begin{enumerate}
\item Sample time points $t^\ell$ uniformly on $[0,T]$ and sample spatial points $x^\ell = (x_1^\ell, \ldots, x_{20}^\ell)$ from the joint distribution of $X_t^1, \ldots, X_t^{20}$ in equation (\ref{SDEx}).
\item  Calculate the error $E^\ell$ at each sampled point $(t^\ell, x^\ell)$ for $\ell =1, \ldots, L$.
\item Aggregate the error over a two-dimensional subspace $\big{(} t^\ell,  \displaystyle \frac{1}{20} \sum_{i=1}^{20} x_i^\ell , E^\ell \big{)}$ for $\ell =1, \ldots, L$.
\item Produce a contour plot from the data $\big{(}  t^\ell,  \frac{1}{20} \sum_{i=1}^{20} x_i^\ell, E^\ell \big{)}_{\ell=1}^L$. The x-axis is $t$ and the y-axis is the geometric average $ \frac{1}{20} \sum_{i=1}^{20} x_i^\ell$, which corresponds to the final condition $g(x)$.
\end{enumerate}
Figure \ref{SubplotFigureContour2} reports both the absolute error and the percent error. The percent error $\frac{ | f(t,x; \theta) - u(t,x) |}{ | u(t,x) | } \times 100\%$ is reported for points where $|u(t,x)| > 0.05$. It should be emphasized that these are \emph{error bounds}; therefore, the actual error could be lower. The contour plot \ref{SubplotFigureContour2} requires significant computations. For each point at which calculate an error bound, a new simulation of (\ref{AmericanOptionBounds}) is required. In total, a large number of simulations are required, which we distribute across hundreds of GPUs on the Blue Waters supercomputer.

\begin{figure}[h!]
\begin{center}
\includegraphics[width=.7\textwidth, height=80mm]{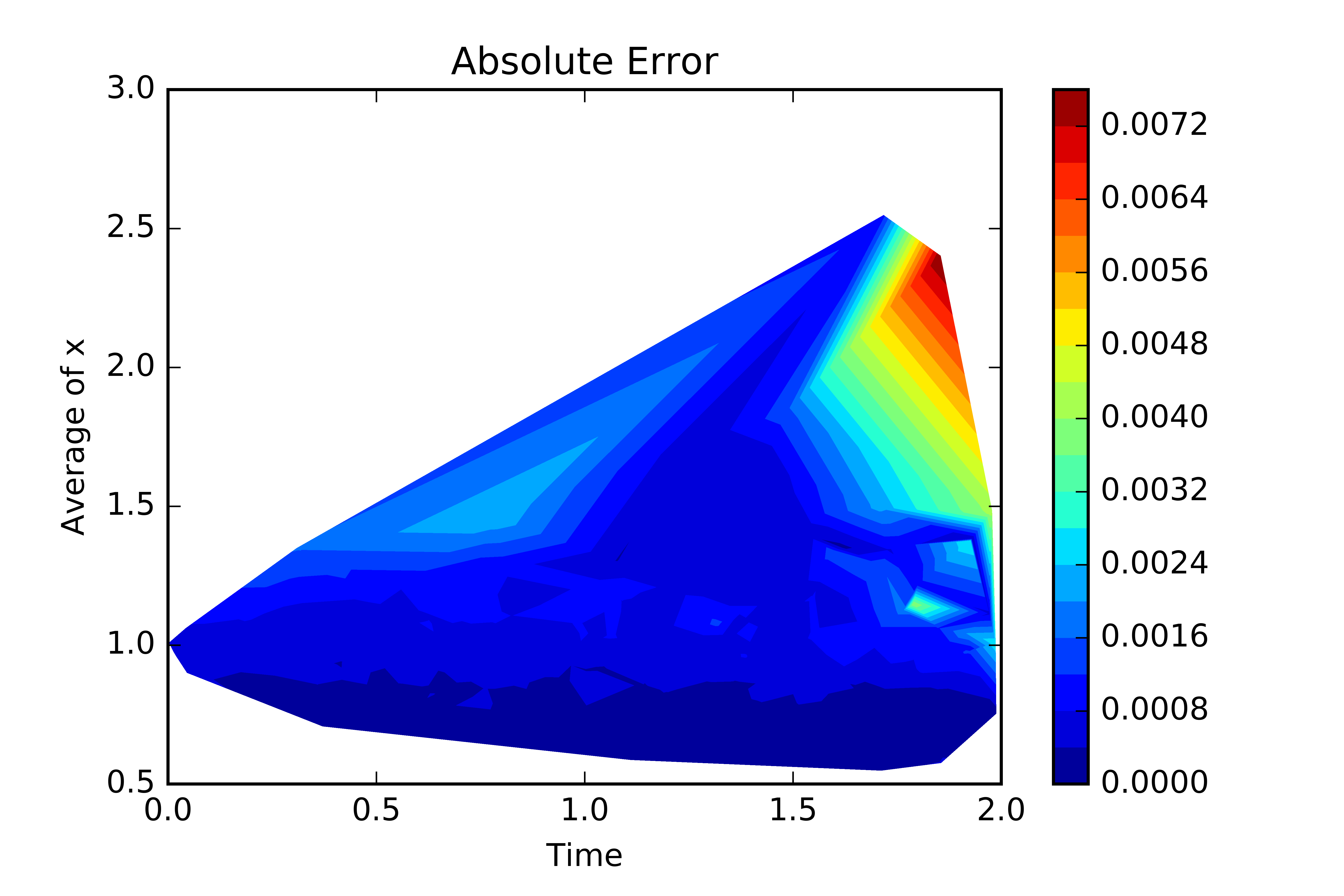}
\includegraphics[width=.7\textwidth, height=80mm]{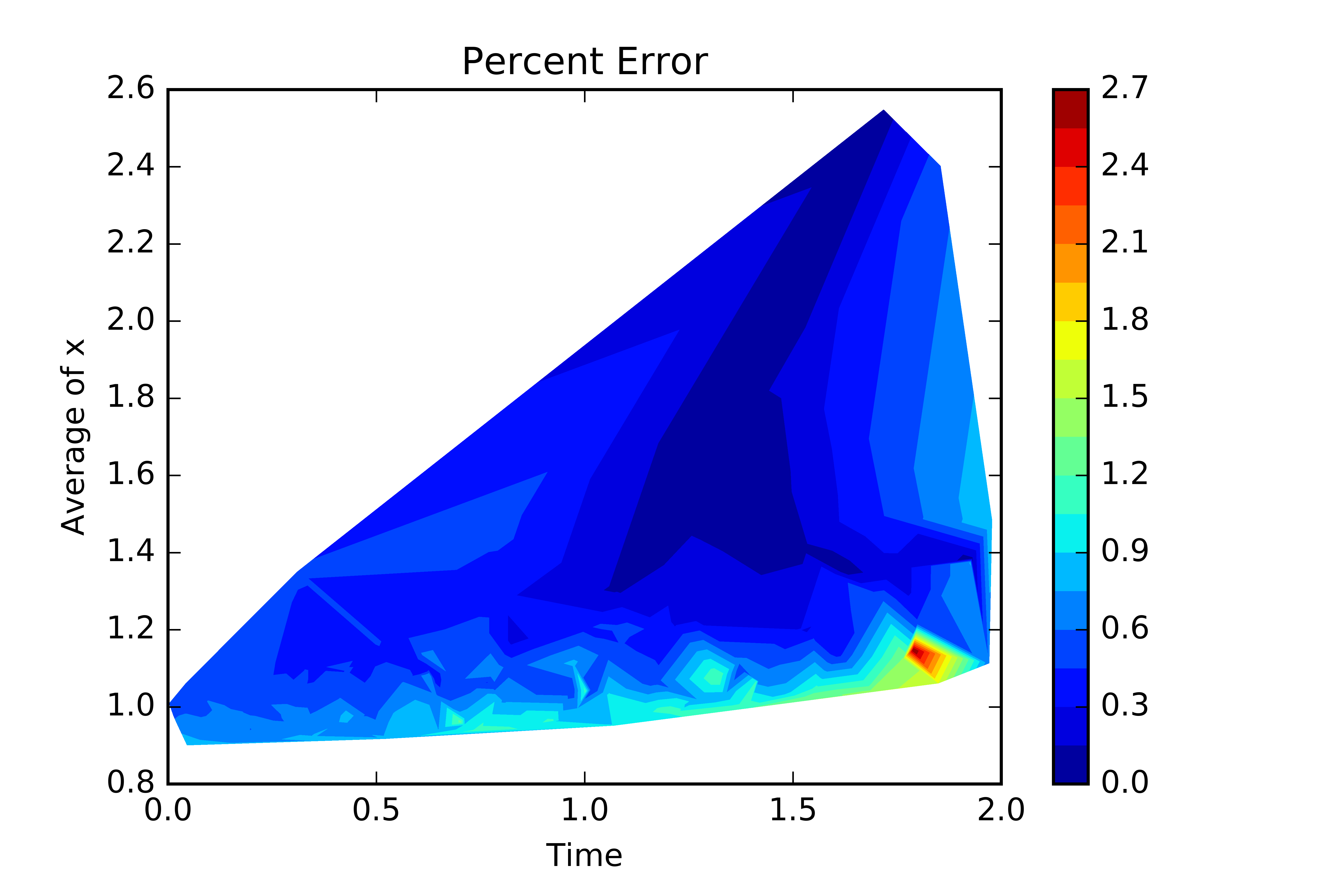}
\end{center}
\caption{Top: Absolute error. Bottom: Percent error. For reference, $u(0, X_0) \in [ 0.10119, 0.10193]$  and the solution at time $T$ is $\max(\textrm{average of x} - 1, 0 )$.}
\label{SubplotFigureContour2}
\end{figure}

\section{High-dimensional Hamilton-Jacobi-Bellman PDE} \label{HJBsection}

We also test the deep learning algorithm on a high-dimensional Hamilton-Jacobi-Bellman (HJB) equation corresponding to the optimal control of a stochastic heat equation. Specifically, we demonstrate that the deep learning algorithm accurately solves the high-dimensional PDE (\ref{HJBpde}). The PDE (\ref{HJBpde}) is motivated by the problem of optimally controlling the stochastic partial differential equation (SPDE):
\begin{eqnarray}
 \frac{\partial v}{\partial t}(t,x) &=& \alpha \frac{\partial^{2} v}{\partial x^{2}}(t,x)  + u(x)  + \sigma \frac{\partial^{2}W}{\partial t\partial x} (t,x), \phantom{....} x \in [0,L],\notag \\
v(t,x=0) &=& \bar v(0), \notag \\
v(t, x=L) &=& \bar v(L), \notag \\
v(t=0,x) &=& v_0(x),
\label{HJBequation0}
\end{eqnarray}
where $u(x)$ is the control and $W(t,x)$ is  \textcolor{black}{ a Brownian sheet (i.e., $\frac{\partial^{2}W}{\partial t\partial x} (t,x)$ is space-time white noise)} defined on a stochastic basis $\left(\Omega,\mathcal{F},\mathcal{F}_{t},\mathbb{P}\right)$.  The
 square integrable, adapted to the filtration $\mathcal{F}_{t}$, control $u$ is a source/sink term which can be used to guide the temperature $v(t,x)$ towards a target profile $\bar v(x)$ on $[0,L]$. As it is discussed in \cite{Cerrai} such problems admit unique solutions in the appropriate generalized sense, see Theorem 3.1 in \cite{Cerrai}. The endpoints at $x=0,L$ are held at the target temperatures. Specifically, the optimal control minimizes
\begin{eqnarray}
\mathbb{E} \bigg{[} \int_0^{\infty} e^{- \gamma s}  \int_0^L \big{[} ( v(s,x) - \bar v (x) )^2 + \lambda u(x)^2  \big{]} dx ds  \bigg{]}.
\label{ObjectiveHJBinfinite}
\end{eqnarray}
 The constant $\gamma > 0$ is a discount factor. The constant $\lambda > 0$ penalizes large values for the control $u(x)$. The goal is to reach the target $\bar v(x)$ while expending the minimum amount of energy. The optimal control $u(x)$ satisfies an \emph{infinite-dimensional} HJB equation. We refer the reader to Theorems 5.3 and 5.4 of  \cite{Cerrai} as well as \cite{Debussche} and \cite{Masiero} for an analysis of infinite-dimensional HJB equations for the stochastic heat equation.

 \textcolor{black}{An example of a problem represented by the SPDE (\ref{HJBequation0}) is the heating of a rod to a target temperature profile. One can control the heat applied to each portion of the rod along its length. There are also random fluctuations in the temperature of the rod due to other environmental factors, which is represented by the  \textcolor{black}{Brownian sheet} $W(t,x)$. The goal is to guide the temperature profile of the rod to the target profile while expending the least amount of energy; see the objective function (\ref{ObjectiveHJBinfinite}).}

 (\ref{HJBequation0}) can be discretized in space, which yields a system of stochastic differential equations (SDEs). \textcolor{black}{(For example, see Section 3.2 of \cite{GainesLectureNotes}.)} This system of SDEs can be used to derive a finite, high-dimensional PDE for the value function and optimal control. \textcolor{black}{That is, we first approximate the SPDE with a finite-dimensional system of SDEs, and then we solve the high-dimensional PDE corresponding to the finite-dimensional system of SDEs.}
\begin{eqnarray}
d X_t^j = \frac{ \alpha}{\Delta^2}   ( X_t^{j+1} - 2 X_t^j + X_t^{j-1} ) dt + U^j_t dt + \frac{ \sigma}{\sqrt{\Delta} }  d W_t^j, \text{ } X^{j}_{0}=v_{0}(j\Delta),
\label{SDEdiscretized}
\end{eqnarray}
where $\Delta$ is the mesh size, $v(t,j \Delta) = X_t^j$, $u(j \Delta) = U^j_t$, and $W_t^j$ are independent standard Brownian motions \textcolor{black}{(see \cite{Gaines}, \cite{Gyongy}, and \cite{GainesLectureNotes} regarding numerical schemes for stochastic parabolic PDEs of the form considered in this section)}. The dimension of the SDE system (\ref{SDEdiscretized}) is $d = \frac{L}{\Delta} - 1$. \textcolor{black}{Note that (\ref{SDEdiscretized}) uses a central difference scheme for the diffusion term in (\ref{HJBequation0}).}

The objective function (\ref{ObjectiveHJBinfinite}) becomes:
\begin{eqnarray}
V( x) = \inf_{U_t \in \mathcal{U} } \mathbb{E} \bigg{[} \int_0^{\infty} e^{- \gamma s} \sum_{j=1}^d \big{[} ( X^j_s - \bar v (j \Delta) )^2 + \lambda (U^j_s)^2  \big{]} \Delta ds \bigg{|} X_0  = x \bigg{]}.
\label{ObjectiveHJB}
\end{eqnarray}
The value function $V(x)$ satisfies a nonlinear PDE with $d$ spatial dimensions $x_1, x_2, \ldots, x_d$.
\begin{eqnarray}
0 &=& \Delta (x - \bar v)^{\top} (x - \bar v) - \frac{1}{4 \lambda \Delta} \sum_{j=1}^d \bigg{(} \frac{\partial V}{\partial x_j}(x) \bigg{)}^2 \notag \\
&+& \frac{\sigma^2}{2 \Delta} \sum_{j=1}^d  \frac{\partial^2 V}{\partial x_j^2} (x) + \frac{ \alpha}{\Delta^2} \sum_{j=1}^d   ( x_{j+1} - 2 x_j + x_{j-1} ) \frac{\partial V}{\partial x_j}(x) - \gamma V(x).
\label{HJBpde}
\end{eqnarray}
The vector $\bar v = ( \bar v ( \Delta), \bar v (2 \Delta), \ldots, \bar v (d \Delta) )$. Note that the values $x_{d+1} = \bar v(L)$ and $x_0 = \bar v(0)$ are constants which correspond to the boundary conditions in (\ref{HJBequation0}). The PDE (\ref{HJBpde}) is high dimensional since the number of dimensions $d = \frac{L}{\Delta} - 1$. The optimal control is
\begin{eqnarray}
U_t^j  = - \frac{1}{2 \lambda \Delta} \frac{\partial V}{\partial x_j}(X_t).
\end{eqnarray}
We solve the PDE (\ref{HJBpde}) using the deep learning algorithm for $d = 21$ dimensions. The size of the domain is $L = 10^{-1}$. The coefficients are $\alpha = 10^{-4}$, $\sigma = 10^{-\frac{1}{2} }$, $\lambda = 1$, and $\gamma = 1$. The target profile is $\bar v(x) = 0$.


The deep learning algorithm's accuracy can be evaluated since a semi-analytic solution is available for (\ref{HJBpde}).\footnote{The PDE (\ref{HJBpde}) has a semi-analytic solution which satisfies a Riccati equation. The Riccati equation can be solved using an iterative method.} Figure \ref{FigureContourHJB} shows a contour plot of the percent error over space. The contour plot is produced in the following way:
\begin{enumerate}
\item Sample spatial points $x^\ell = (x_1^\ell, \ldots, x_{21}^\ell)$ from the distribution of (\ref{SDEdiscretized}) for $\ell =1, \ldots, L$.
\item  Calculate the percent error at each sampled point. The percent error is $A^\ell = \frac{ | f(x^\ell; \theta) - V(x^\ell) |}{ | V(x^\ell) | } \times 100\%$.
\item Aggregate the accuracy over a two-dimensional subspace $\big{(} x_{11}^\ell, \frac{1}{21} \displaystyle \sum_{i=1}^{21} x_i^\ell, A^\ell \big{)}$ for $\ell =1, \ldots, L$.
\item Produce a contour plot from the data $\big{(} x_{11}^\ell, \frac{1}{21} \displaystyle \sum_{i=1}^{21} x_i^\ell, A^\ell \big{)}_{\ell=1}^L$. The x-axis is $x_{11}$ and the y-axis is the average $\frac{1}{21} \displaystyle \sum_{i=1}^{21} x_i$. This corresponds to $v(t,x)$ at the midpoint $x = \frac{L}{2}$ and the average $\frac{1}{L} \int_0^L v(t,x) dx$, respectively.
\end{enumerate}
The average percent error over the entire space is $0.1 \%$.

\begin{figure}[h!]
\begin{center}
\includegraphics[width=.7\textwidth, height=80mm]{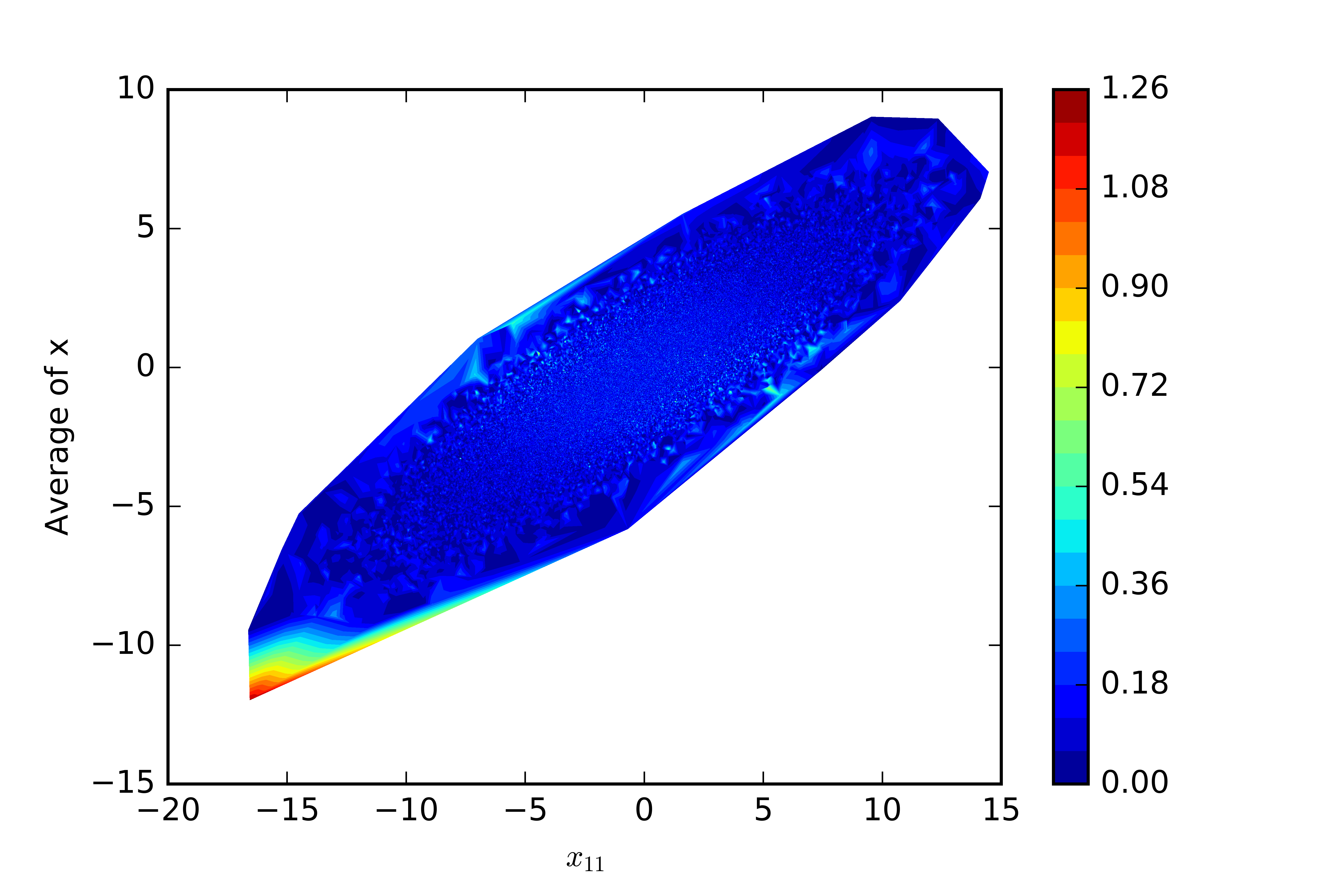}
\end{center}
\caption{Contour plot of the percent error for the deep learning algorithm for a 21-dimensional Hamilton-Jacobi-Bellman PDE. The horizontal axis is the $11$-th dimension. The vertical axis is the average of all dimensions.}
\label{FigureContourHJB}
\end{figure}

\textcolor{black}{Lastly, we close this section by mentioning that in the recent paper \cite{Weinan} (see also \cite{Jentzen}) the authors develop a machine learning algorithm that provides the value at a single point in time and space of the solution to a class of HJB equations  which admit explicit solution that can be obtained through  the  Cole-Hopf  transformation. Their method relies on characterizing the solution via backward stochastic differential equations (BSDE). In contrast, the current work (a) does not rely on BSDE type representations through nonlinear Feynman-Kac formulas, and (b) allows to recover the whole object (i.e. the solution across all points in time and space).}

\section{Burgers' equation} \label{BurgerEquation}

It is often of interest to find the solution of a PDE over a range of problem setups (e.g., different physical conditions and boundary conditions). For example, this may be useful for the design of engineering systems or uncertainty quantification. The problem setup space may be high-dimensional and therefore may require solving many PDEs for many different problem setups, which can be computationally expensive.

Let the variable $p$ represent the problem setup (i.e., physical conditions, boundary conditions, and initial conditions). The variable $p$ takes values in the space $\mathcal{P}$, and we are interested in the solution of the PDE $u(t,x; p)$. (This is sometimes called a ``parameterized class of PDEs".) In particular, suppose $u(t,x; p)$ satisfies the PDE
\begin{eqnarray}
\frac{\partial u}{\partial t}(t,x; p) &=& \mathcal{L}_p u(t,x; p), \phantom{....} (t,x) \in [0,T] \times \Omega, \notag \\
u(t,x;p) &=& g_p(x),  \phantom{....} (t,x) \in [0,T] \times \partial \Omega, \notag \\
u(t=0,x;p) &=& h_p(x), \phantom{....} x \in \Omega.
\label{PDEu_p}
\end{eqnarray}

A traditional approach would be to discretize the $\mathcal{P}$-space and re-solve the PDE many times for many different points $p$. However, the total number of grid points (and therefore the number of PDEs that must be solved) grows exponentially with the number of dimensions, and $\mathcal{P}$ is typically high-dimensional.

We propose to use the DGM algorithm to approximate the \emph{general solution} to the PDE (\ref{PDEu_p}) for different boundary conditions, initial conditions, and physical conditions. The deep neural network is trained using stochastic gradient descent on a sequence of random time, space, and problem setup points $(t, x, p)$. Similar to before,

\begin{itemize}
\item Initialize $\theta$.
\item Repeat until convergence:
\begin{itemize}
\item Generate random samples $(t, x, p)$ from $[0,T] \times \Omega \times \mathcal{P}$,  $(\tilde t, \tilde x)$ from $[0,T] \times \partial \Omega$, and
$\hat x$ from $\Omega$.
\item Construct the objective function
\begin{eqnarray}
J(\theta) &=& \bigg{(} \frac{\partial f}{\partial t}(t,x,p; \theta) -  \mathcal{L}_p f(t,x,p;  \theta) \bigg{)}^2 \notag \\
&+& \bigg{(} g_p(\tilde x) - f( \tilde t,\tilde x,p; \theta) \bigg{)}^2 \notag \\
&+& \bigg{(} h_p(\hat x) - f( 0,\hat x,p; \theta) \bigg{)}^2.
\end{eqnarray}
\item Update $\theta$ with a stochastic gradient descent step
\begin{eqnarray}
\theta \longrightarrow \theta - \alpha \nabla_{\theta} J( \theta ),
\end{eqnarray}
where $\alpha$ is the learning rate.
\end{itemize}

\end{itemize}
If $x$ is low-dimensional ($d \leq 3$), which is common in many physical PDEs, the first and second partial derivatives of $f$ can be calculated via chain rule or approximated by finite difference. We implement our algorithm for Burgers' equation on a finite domain.

\begin{eqnarray}
&\phantom{.}& \frac{\partial u}{\partial t} =  \nu \frac{\partial^2 u}{\partial x^2} - \alpha u \frac{\partial u}{\partial x}, \phantom{....} (t,x) \in [0,1] \times [0,1], \notag \\
 &\phantom{.}&  u(t, x =0) = a, \notag \\
 &\phantom{.}& u(t, x=1)  = b, \notag \\
 &\phantom{.}& u(t=0, x) = g(x), \phantom{....} x \in [0,1]. \notag
\end{eqnarray}

The problem setup space is $\mathcal{P} = (\nu, \alpha, a, b) \in \mathbb{R}^4$. The initial condition $g(x)$ is chosen to be a linear function which matches the boundary conditions $u(t, x =0) = a$ and $u(t, x =1) = b$. We train a \emph{single neural network} to approximate the solution of $u(t, x; p)$ over the entire space $(t,x, \nu, \alpha, a, b) \in [0,1] \times [0,1] \times [10^{-2}, 10^{-1} ] \times [10^{-2}, 1 ] \times [-1, 1] \times [-1, 1]$. We use a larger network ($6$ layers, $200$ units per layer) than in the previous numerical examples. Figure \ref{SubplotBurger} compares the deep learning solution with the exact solution for several different problem setups $p$. The solutions are very close; in several cases, the two solutions are visibly indistinguishable. The deep learning algorithm is able to accurately capture the shock layers and boundary layers.

Figure \ref{SubplotBurger2} presents the accuracy of the deep learning algorithm for different times $t$ and different choices of $\nu$. As $\nu$ becomes smaller, the solution becomes steeper. It also shows the shock layer forming over time. The contour plot (\ref{BurgerContourPlot}) reports the absolute error of the deep learning solution for different choices of $b$ and $\nu$.

\begin{figure}[h!]
\begin{center}
\includegraphics[width=.4\textwidth, height=50mm]{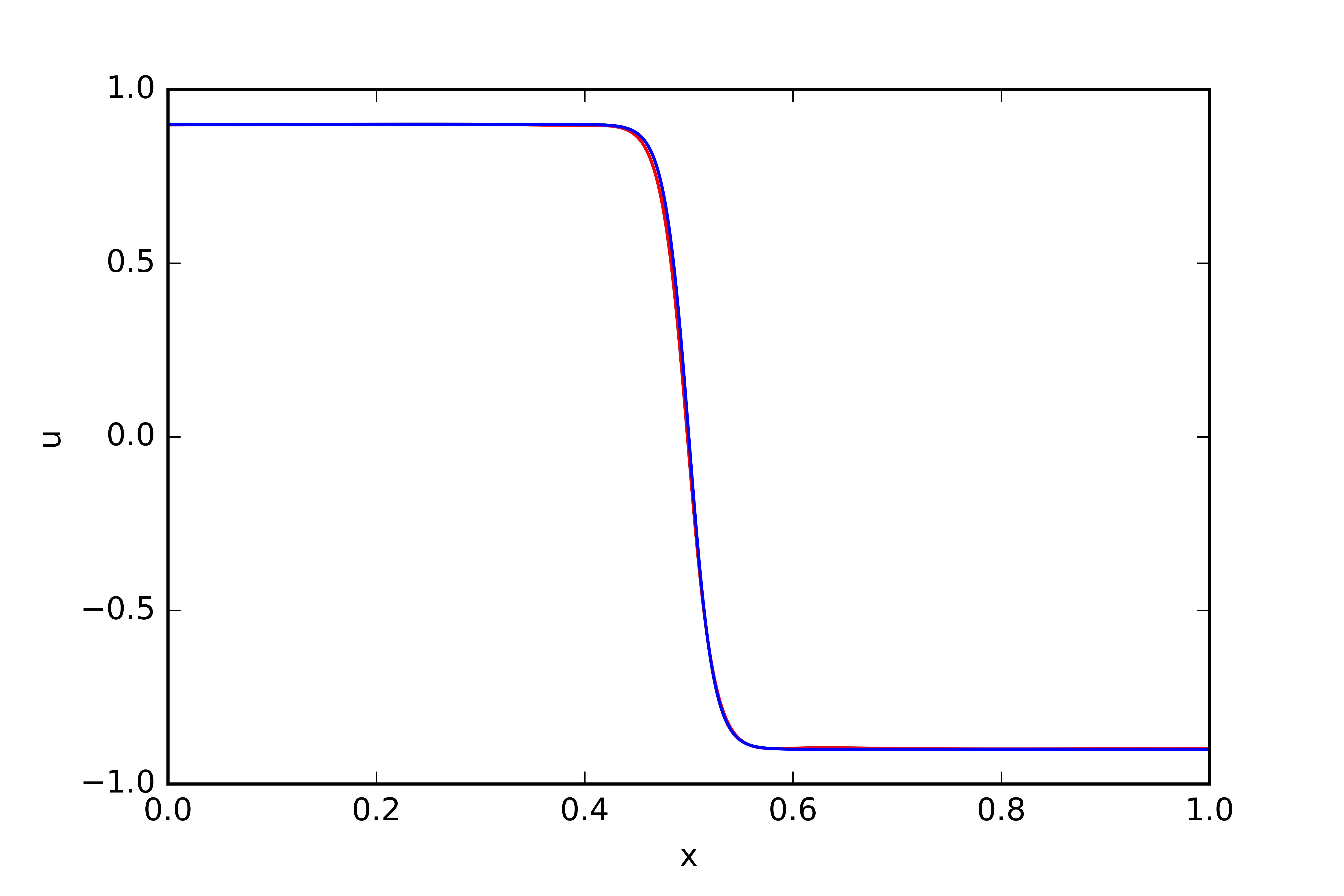}
\includegraphics[width=.4\textwidth, height=50mm]{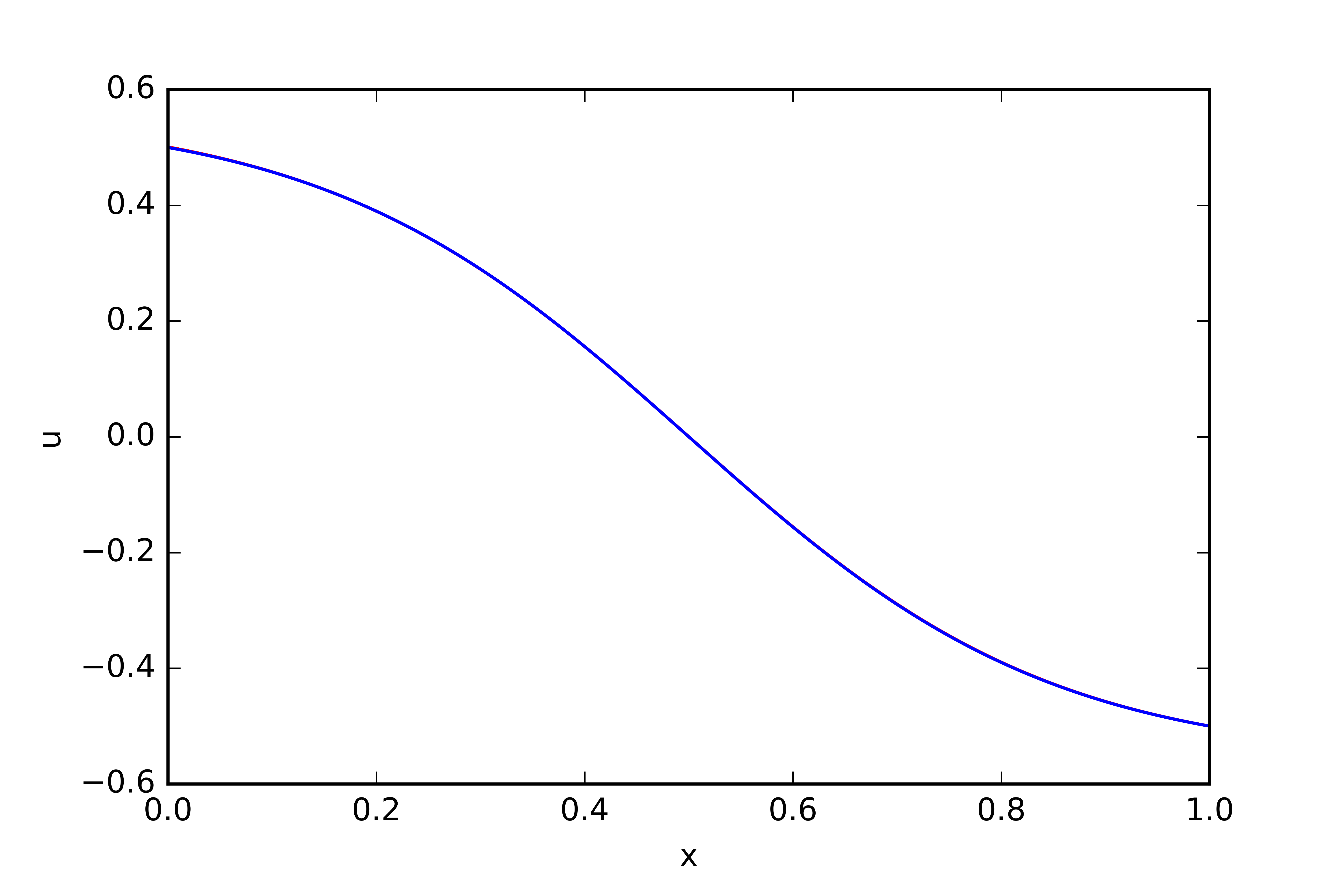}
\includegraphics[width=.4\textwidth, height=50mm]{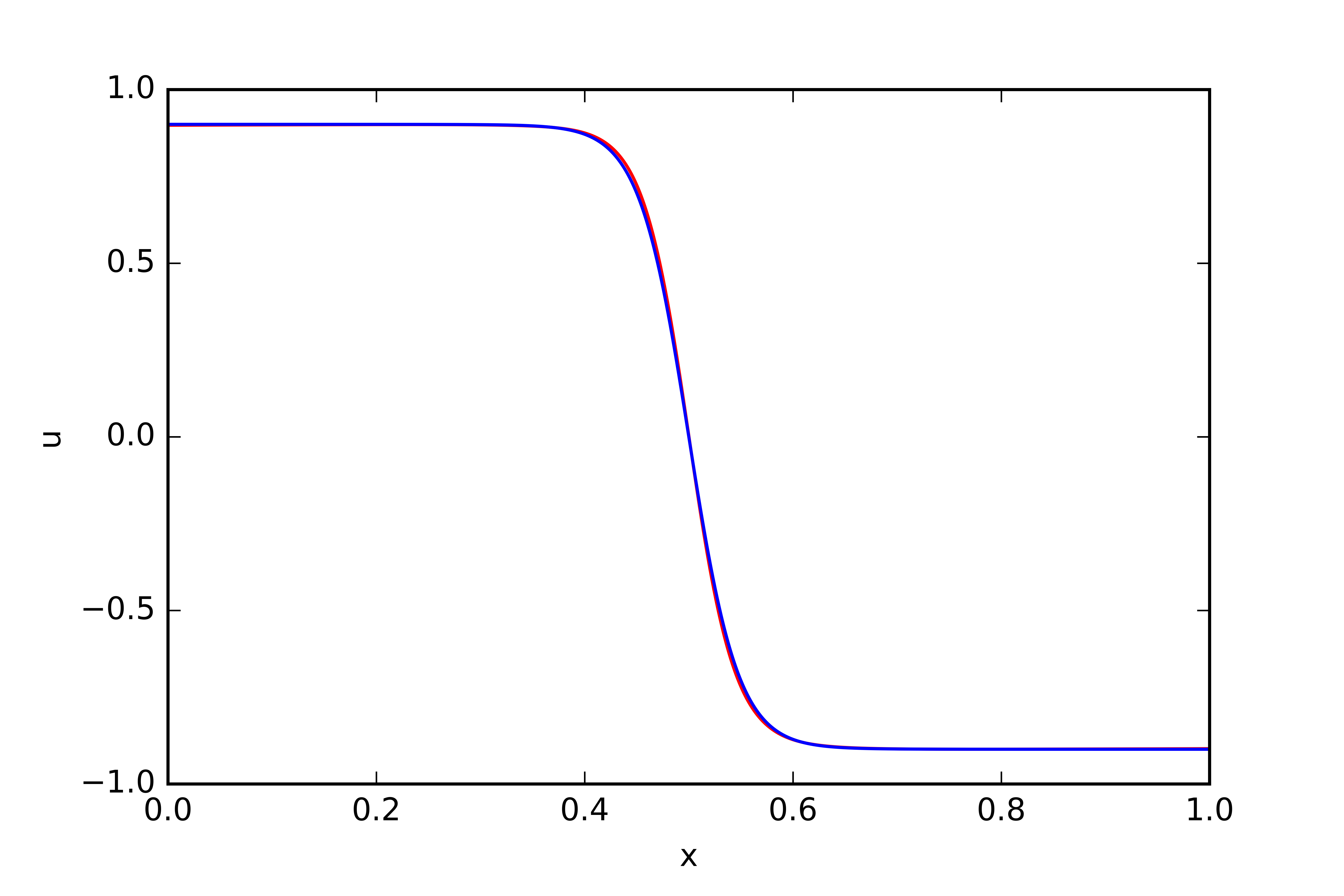}
\includegraphics[width=.4\textwidth, height=50mm]{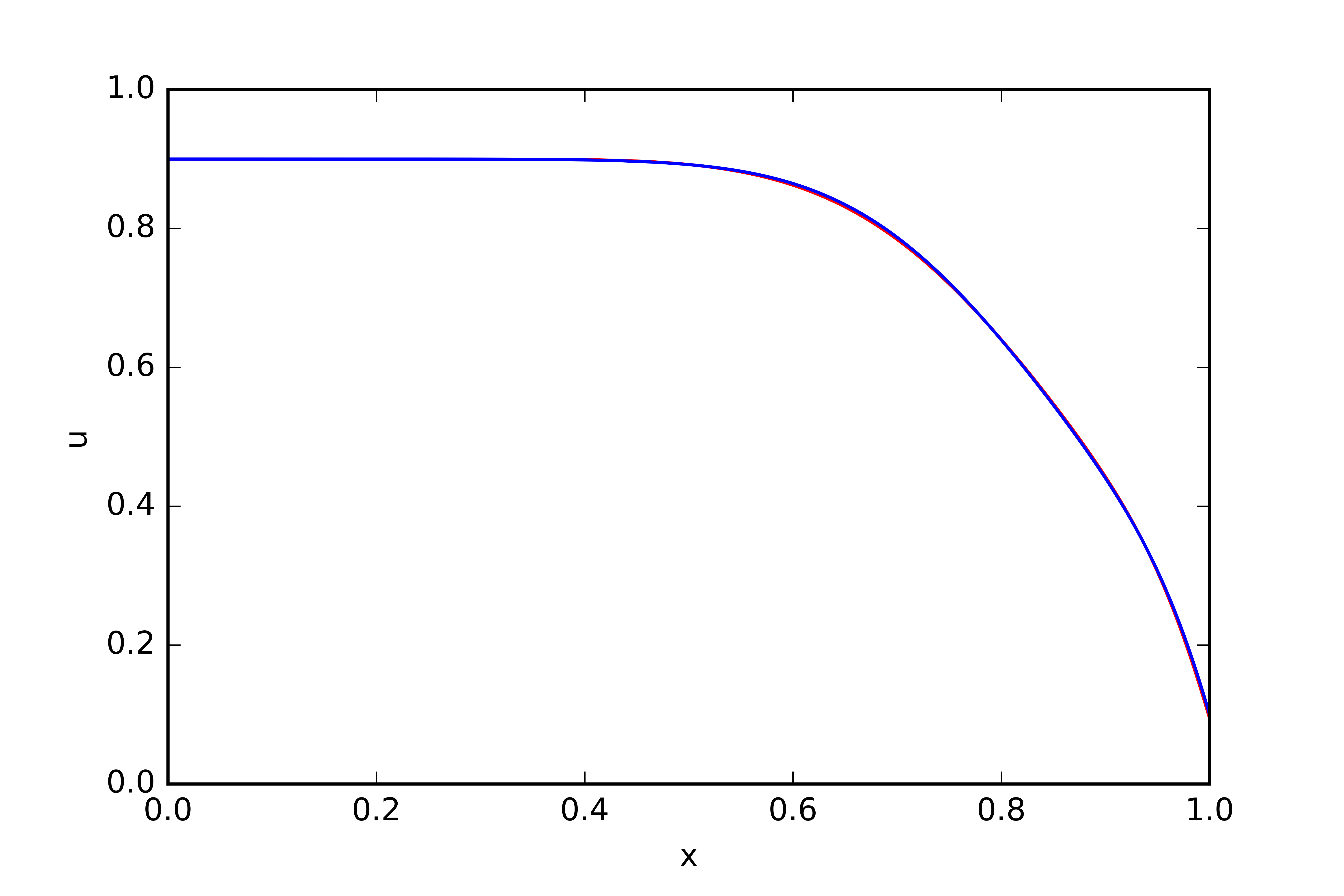}
\includegraphics[width=.4\textwidth, height=50mm]{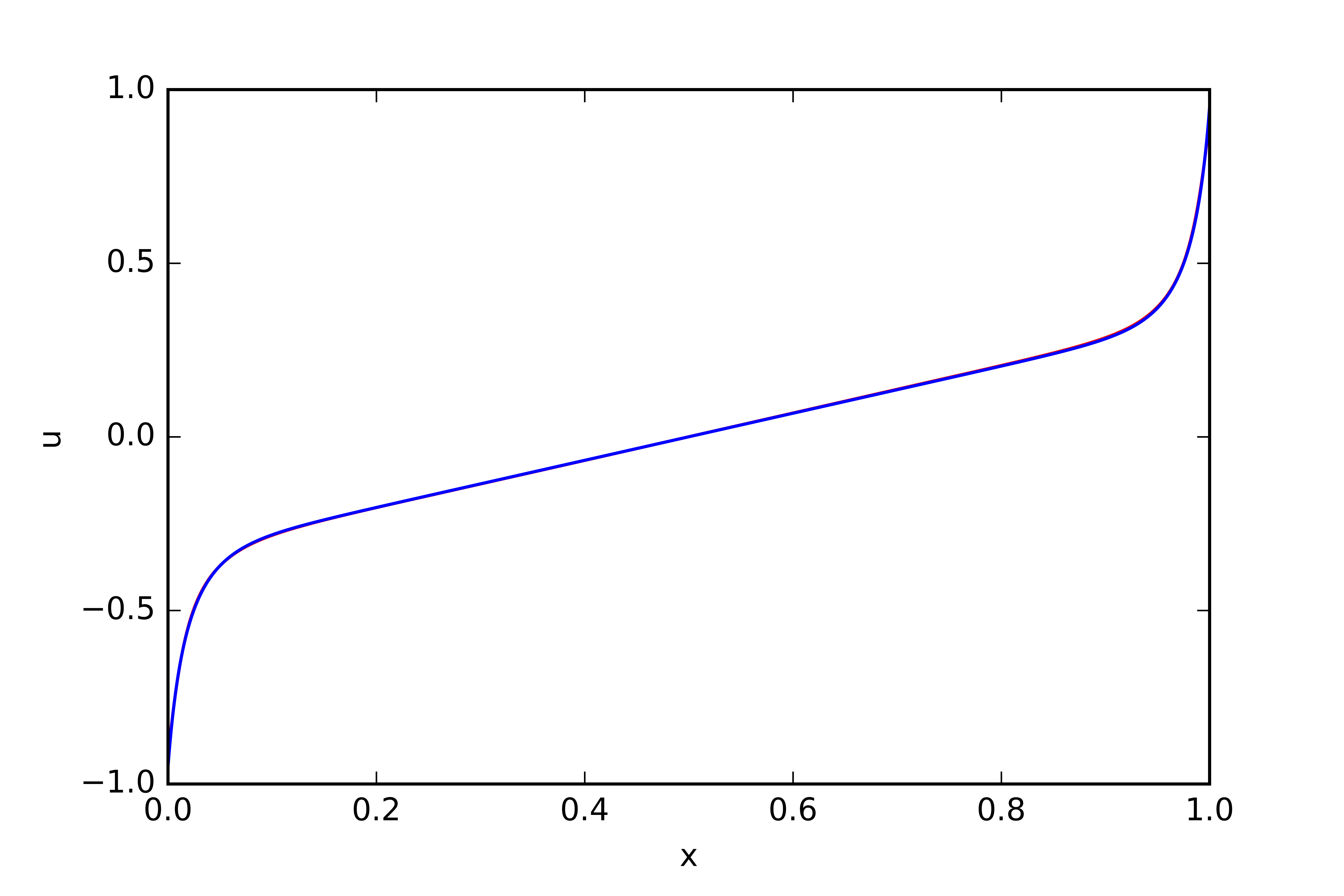}
\includegraphics[width=.4\textwidth, height=50mm]{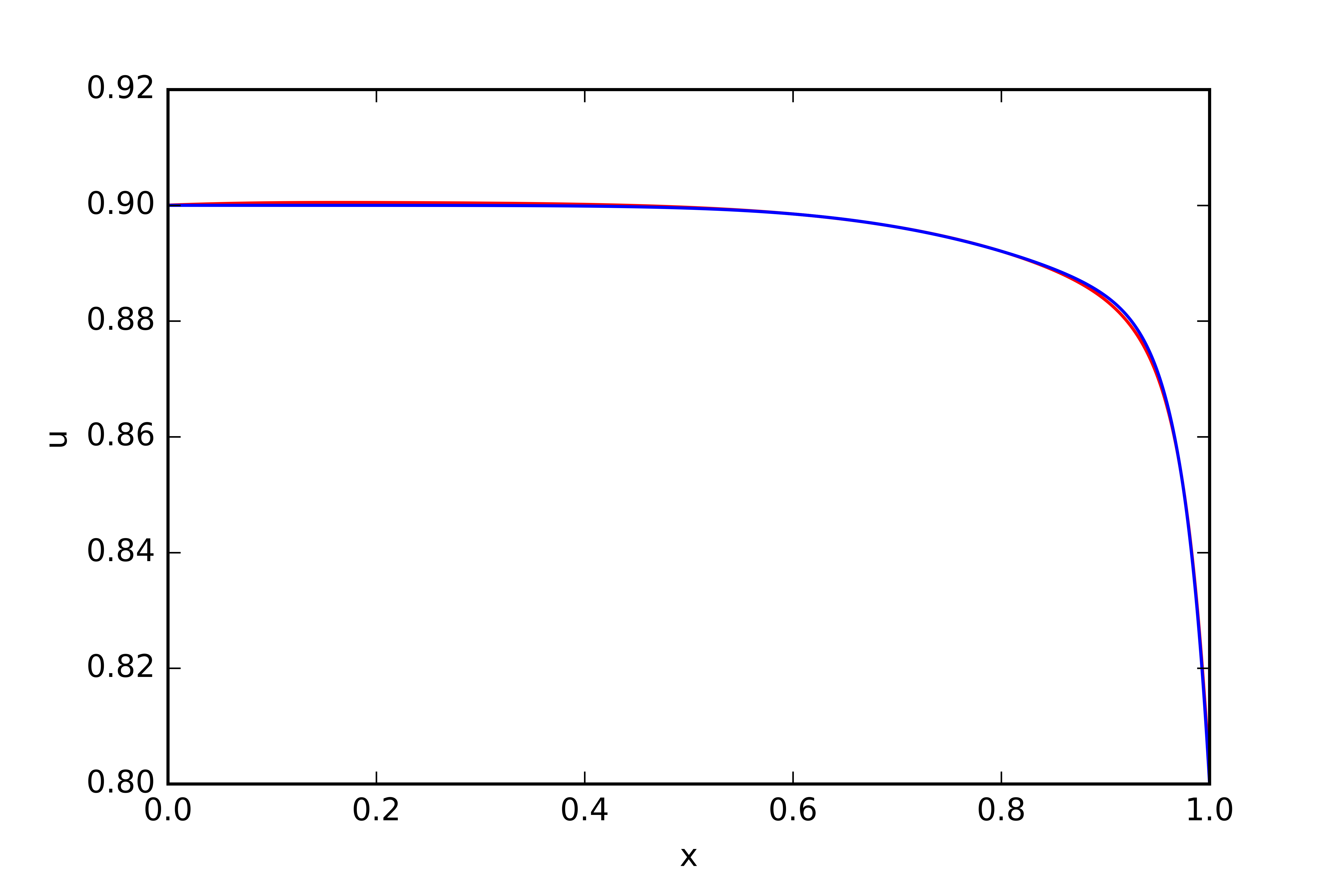}
\end{center}
\caption{The deep learning solution is in \textcolor{red}{red}. The ``exact solution", found via finite difference, is in \textcolor{blue}{blue}. Solutions are reported at time $t =1$. The solutions are very close; in several cases, the two solutions are visibly indistinguishably. The problem setups, in counter-clockwise order, are $( \nu, \alpha, a,b) =
(0.01, 0.95, 0.9, - 0.9),    (0.02, 0.95, 0.9, -0.9),  (0.01, 0.95, -0.95, 0.95),  (0.02, 0.9, 0.9, 0.8),  (0.01, 0.75, 0.9, 0.1),$ and $(0.09, 0.95, 0.5 -0.5)$.}
\label{SubplotBurger}
\end{figure}

\begin{figure}[h!]
\begin{center}
\includegraphics[width=.4\textwidth, height=50mm]{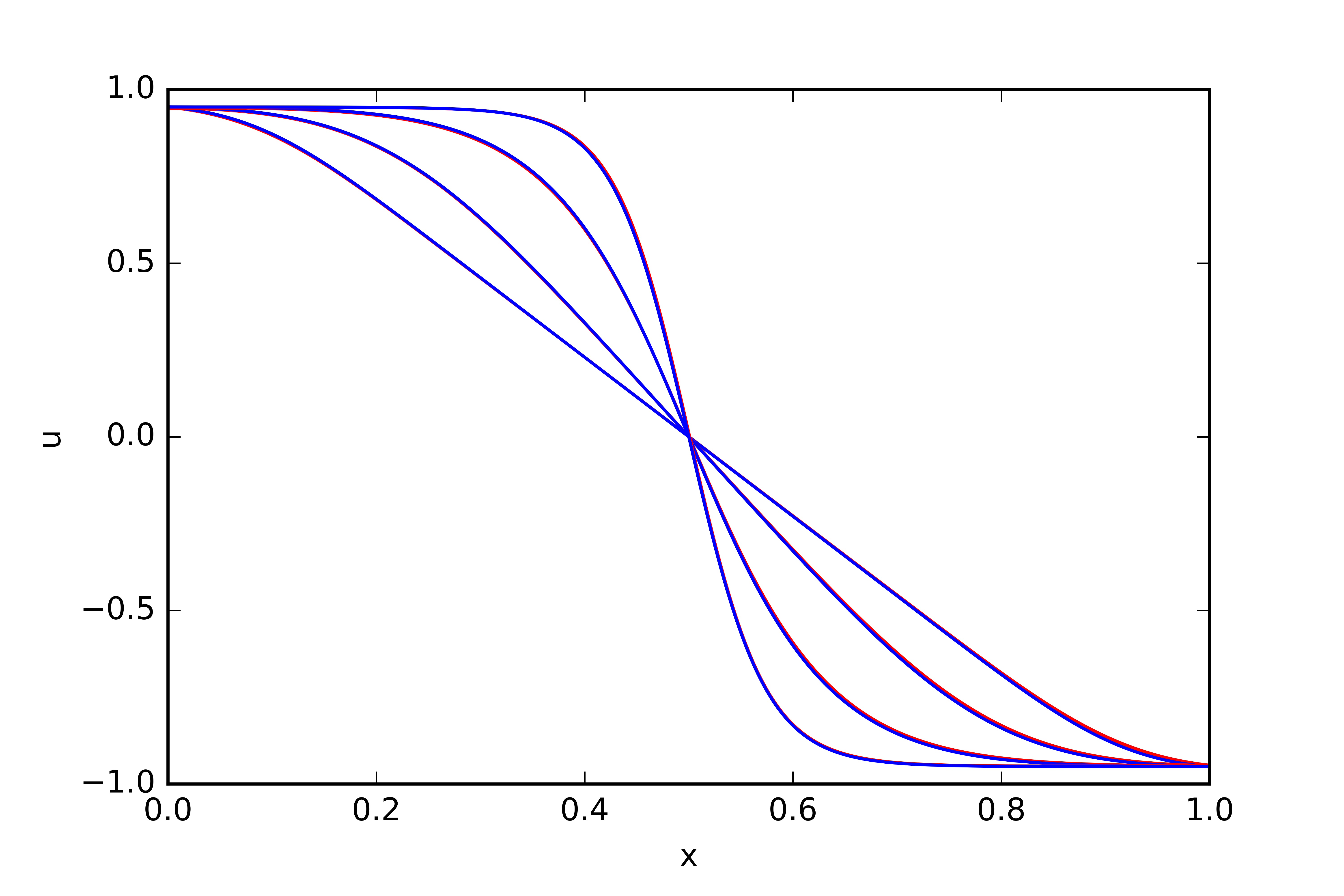}
\includegraphics[width=.4\textwidth, height=50mm]{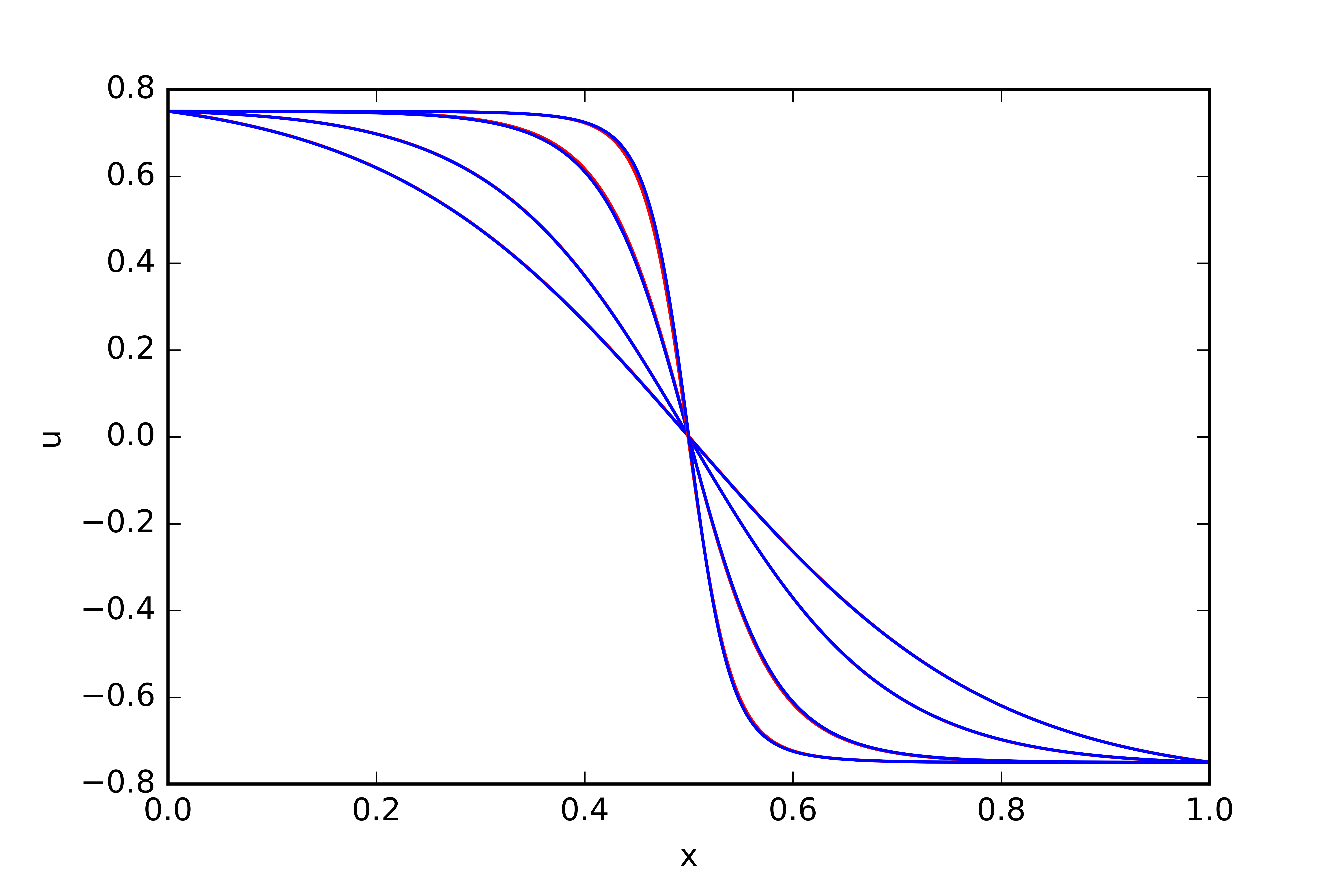}
\end{center}
\caption{The deep learning solution is in \textcolor{red}{red}. The ``exact solution", found via finite difference, is in \textcolor{blue}{blue}. \textbf{Left plot:} Comparison of solutions at times $t = 0.1, 0.25, 0.5, 1$ for $( \nu, \alpha, a,b) = (0.03, 0.9, 0.95, -0.95)$. \textbf{Right plot:} Comparison of solutions for $\nu =  0.01, 0.02, 0.05, 0.09$ at time $t =1$ and with $(\alpha, a, b) = (0.8, 0.75, - 0.75)$.}
\label{SubplotBurger2}
\end{figure}



\begin{figure}[h!]
\begin{center}
\includegraphics[width=.7\textwidth, height=80mm]{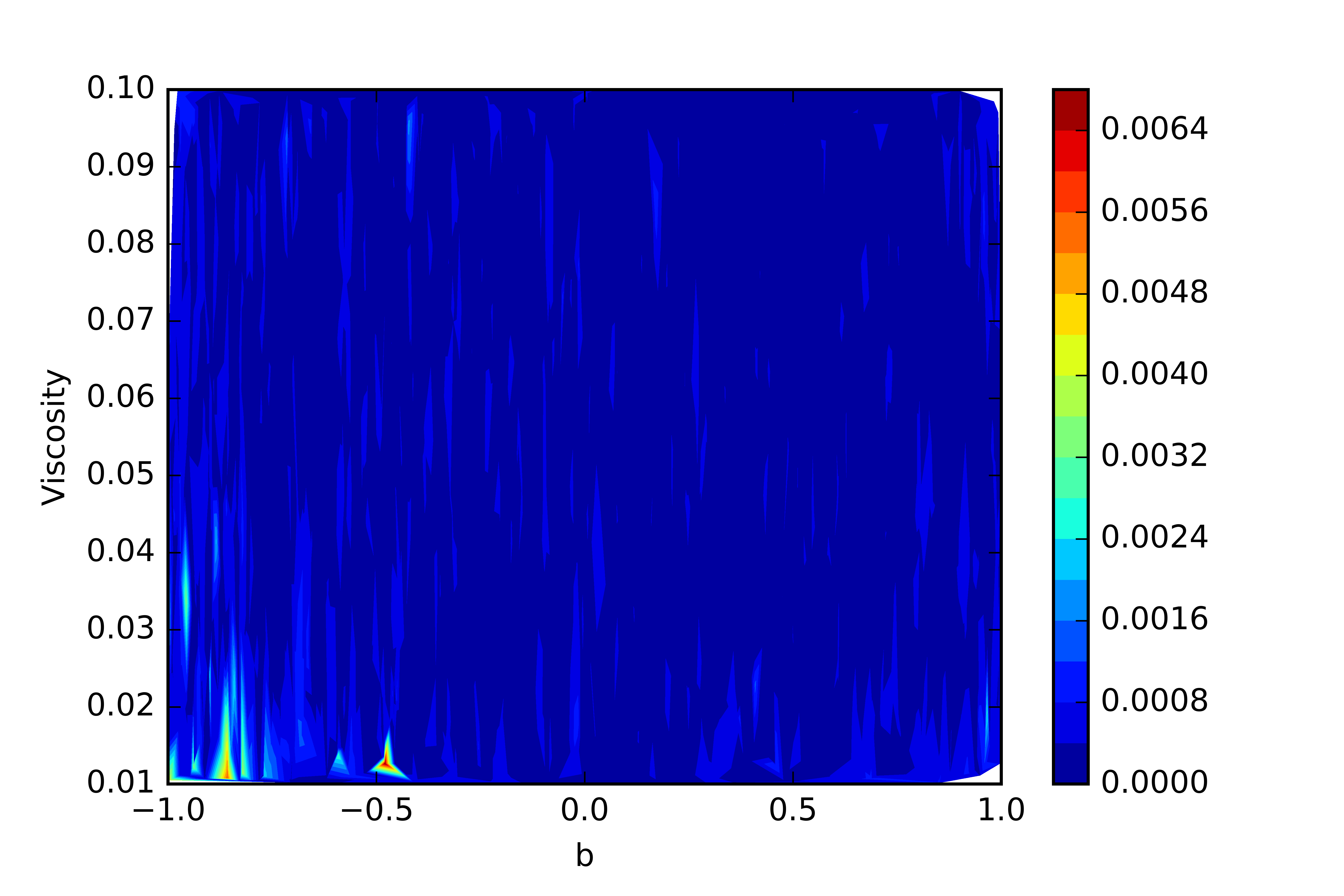}
\end{center}
\caption{Contour plot of the average absolute error of the deep learning solution for different $b$ and $\nu$ (the viscosity). The absolute error is averaged across $x \in [0,1]$ for time $t = 1$.}
\label{BurgerContourPlot}
\end{figure}

\section{Neural Network Approximation Theorem for PDEs}  \label{ApproximationProof}
Let the L$^2$ error $J(f)$ measure how well the neural network $f$ satisfies the differential operator, boundary condition, and initial condition. Define $\mathfrak{C}^{n}$ as the class of neural networks with $n$ hidden units and let $f^{n}$ be a neural network with $n$ hidden units which minimizes $J(f)$. We prove that
\begin{eqnarray}
&& \text{there exists }f^{n}\in \mathfrak{C}^{n} \text{ such that }  J(f^{n})\rightarrow 0, \text{ as }n\rightarrow \infty, \text{ and} \notag \\
&& f^{n} \rightarrow u \phantom{....} \textrm{as} \phantom{....} n \rightarrow \infty, \notag
\end{eqnarray}
in the appropriate sense, for \textcolor{black}{a class of quasilinear parabolic PDEs} with the principle term in divergence form under certain growth and smoothness assumptions on the nonlinear terms. Our theoretical result only covers \textcolor{black}{a class of quasilinear parabolic PDEs} as described in this section. However, the numerical results of this paper indicate that the results are more broadly applicable.
%

The proof requires the joint analysis of the approximation power of neural networks as well as the continuity properties of partial differential equations. First, we show that the neural network can satisfy the differential operator, boundary condition, and initial condition arbitrarily well for sufficiently large $n$.
\begin{eqnarray}
J(f^{n}) \rightarrow 0 \phantom{.....} \textrm{as} \phantom{.....} n \rightarrow \infty.
\label{errorJtheorem0}
\end{eqnarray}
Let $u$ be the solution to the PDE. The statement (\ref{errorJtheorem0}) does not necessarily imply that $f^{n} \rightarrow u$. One challenge to proving convergence is that we only have $L^{2}$ control of the error.  We prove convergence for the case of homogeneous boundary data, i.e., $g(t,x)=0$, by first establishing that each neural network $\{ f^{n} \}_{n=1}^{\infty}$ satisfies a PDE with a source term $h^n(t,x)$.  Importantly, the source terms $h^n(t,x)$ are only known to be vanishing in $L^2$. We are then able to prove that the convergence of $f^{n} \rightarrow u$ as $n \rightarrow \infty$ in the appropriate space holds using  compactness arguments.

The precise statement of the theorem and the presentation of the proof is in the next two sections. Section \ref{errorJanalysis} proves that $J(f^{n}) \rightarrow 0$ as $n \rightarrow \infty$. Section \ref{ConvergenceProof} contains convergence results of $f^{n}$ to the solution $u$ of the PDE as $n \rightarrow \infty$. The main result is Theorem \ref{T:MainConvTheorem}. \textcolor{black}{For readability purposes the corresponding proofs are in Appendix \ref{S:Proofs}.}

\subsection{Convergence of the L$^2$ error $J(f)$} \label{errorJanalysis}
In this section, we present a theorem guaranteeing the existence of multilayer feed forward networks $f$ able to universally approximate solutions of quasilinear \textcolor{black}{parabolic} PDEs in the sense that there is $f$ that makes the objective function $J(f)$  arbitrarily small. To do so, we use the results of \cite{Hornik} on universal approximation of functions and their derivatives and make appropriate assumptions on the coefficients of the PDEs to guarantee that a classical solution exists (since then the results of \cite{Hornik} apply).

Consider a bounded set $\Omega\subset \mathbb{R}^{d}$ with a smooth boundary $\partial \Omega$ and denote $\Omega_{T}=(0,T]\times \Omega$ and $\partial\Omega_{T}=(0,T]\times\partial\Omega$. \textcolor{black}{In this subsection} we consider the class of quasilinear parabolic PDE's of the form 

\begin{align}
\partial_{t} u(t,x)&-\text{div} \left(\alpha(t,x,u(t,x), \nabla u(t,x))\right)+\gamma(t,x,u(t,x),\nabla u(t,x))=0, \textrm{ for }(t,x)\in\Omega_{T}\nonumber\\
u(0,x)&=u_{0}(x), \textrm{ for }x\in\Omega\nonumber\\
u(t,x)&=g(t,x), \textrm{ for }(t,x)\in\partial\Omega_{T}\label{Eq:QuasiLinearParabolicPDE}
\end{align}

For notational convenience, let us write the operator of (\ref{Eq:QuasiLinearParabolicPDE}) as $\mathcal{G}$. Namely, let us denote
\begin{align*}
\mathcal{G}[u](t,x)&=\partial_{t} u(t,x)-\text{div} \left(\alpha(t,x,u(t,x), \nabla u(t,x))\right)+\gamma(t,x,u(t,x),\nabla u(t,x)).
\end{align*}

Notice that we can write
\begin{align*}
\mathcal{G}[u](t,x)&=\partial_{t} u(t,x)-\sum_{i,j=1}^{d}\frac{\partial\alpha_{i}(t,x,u(t,x), \nabla u(t,x))}{\partial u_{x_{j}}}\partial_{x_{i},x_{j}}u(t,x)+\hat{\gamma}(t,x,u(t,x),\nabla u(t,x)),
\end{align*}
where
\begin{align*}
\hat{\gamma}(t,x,u,p)=\gamma(t,x,u,p)-\sum_{i=1}^{d}\frac{\partial\alpha_{i}(t,x,u,p)}{\partial u}\partial_{x_{i}}u -\sum_{i=1}^{d}\frac{\partial\alpha_{i}(t,x,u,p)}{\partial x_{i}}.
\end{align*}

For the purposes of this section, we consider equations of the type (\ref{Eq:QuasiLinearParabolicPDE}) that have classical solutions. 

In particular we assume that there is a unique $u(t,x)$ solving (\ref{Eq:QuasiLinearParabolicPDE}) such that
\begin{align}
u(t,x)\in \mathcal{C}(\bar{\Omega}_{T})\bigcap \mathcal{C}^{1+\eta/2,2+\eta}(\Omega_{T}) \textrm{ with }\eta\in(0,1) \textrm{ and that } \sup_{(t,x)\in\Omega_{T}}\sum_{k=1}^{2}|\nabla_{x}^{(k)}u(t,x)|<\infty.\label{Eq:UniversalApproxRegularityPDE}
\end{align}

We refer the interested reader to Theorems  5.4, 6.1 and 6.2 of Chapter V in \cite{Ladyzhenskaya} for specific general conditions on $\alpha,\gamma$ guaranteeing the validity of the aforementioned statement.

Universal approximation results for single functions and their derivatives have been obtained under various assumptions in \cite{Cybenko,HornikEtAl,Hornik}. In this paper, we use \textcolor{black}{Theorem 3 of} \cite{Hornik}. Let us recall the setup appropriately modified for our case of interest. Let $\psi$ be an activation function, e.g., of sigmoid type, of the hidden units and define the set
\begin{align}
\mathfrak{C}^{n}(\psi)=\left\{\zeta(t,x):\mathbb{R}^{1+d}\mapsto\mathbb{R}: \zeta(t,x)=\sum_{i=1}^{n}\beta_{i}\psi\left(\alpha_{1,i}t+\sum_{j=1}^{d}\alpha_{j,i}x_{j}+c_{j}\right)\right\}.\label{Eq:FeedForwardNN}
\end{align}
 \textcolor{black}{where $\theta=\left(\beta_{1},\cdots,\beta_{n},\alpha_{1,1},\cdots,\alpha_{d,n},c_{1},c_{1},\cdots, c_{n}\right)\in\mathbb{R}^{2n+n(1+d)}$ compose the elements of the parameter space}.
Then we have the following result.

\begin{theorem}\label{T:UniversalApproximationThmPDE}
Let $\mathfrak{C}^{n}(\psi)$ be given by (\ref{Eq:FeedForwardNN}) where $\psi$ is assumed to be in $\mathcal{C}^{2}(\mathbb{R}^{d})$, bounded and non-constant. Set $\mathfrak{C}(\psi)=\bigcup_{n\geq 1}\mathfrak{C}^{n}(\psi)$. Assume that $\Omega_{T}$ is compact and consider the measures $\nu_{1}$, $\nu_{2}$, $\nu_{3}$ whose support is contained in $\Omega_{T}$, $\Omega$ and $\partial \Omega_{T}$ respectively. In addition, assume that the PDE (\ref{Eq:QuasiLinearParabolicPDE}) has a unique classical solution such that (\ref{Eq:UniversalApproxRegularityPDE}) holds. Also, assume that the nonlinear terms $\frac{\partial\alpha_{i}(t,x,u, p)}{\partial p_{j}}$ and $\hat{\gamma}(t,x,u,p)$ are locally Lipschitz in $(u,p)$ with Lipschitz constant that can have at most polynomial growth on $u$ and $p$, uniformly with respect to $t,x$. Then, for every $\epsilon>0$, there exists a positive constant $K>0$ that may depend on 
$\sup_{\Omega_{T}}|u|$, $\sup_{\Omega_{T}}|\nabla_{x}u|$ and $\sup_{\Omega_{T}}|\nabla^{(2)}_{x}u|$
such that there exists a function $f\in\mathfrak{C}(\psi)$ that satisfies
\[
J(f)\leq K\epsilon.
\]
\end{theorem}

The proof of this theorem is in the Appendix.

\subsection{Convergence of the neural network to the PDE solution} \label{ConvergenceProof}

We now prove, under stronger conditions, the convergence of the neural networks $f^{n}$ to the solution $u$ of the PDE
\begin{align}
\partial_{t} u(t,x)&-\text{div} \left(\alpha(t,x,u(t,x), \nabla u(t,x))\right)+\gamma(t,x,u(t,x),\nabla u(t,x))=0, \textrm{ for }(t,x)\in\Omega_{T}\nonumber\\
u(0,x)&=u_{0}(x), \textrm{ for }x\in\Omega\nonumber\\
u(t,x)&=0, \textrm{ for }(t,x)\in\partial\Omega_{T},\label{Eq:QuasiLinearParabolicPDE2}
\end{align}
as $n \rightarrow \infty$. Notice that we have restricted the discussion to homogeneous boundary data. We do this for both presentation and mathematical reasons. \footnote{We set $u(t,x)=0, \textrm{ for }(t,x)\in\partial\Omega_{T}$, i.e., $g=0$, to circumvent certain technical difficulties arising due to inhomogeneous boundary conditions. If $g\neq0$ such that $g$ is the trace of some appropriately smooth function, say $\phi$, then one can reduce the inhomogeneous boundary conditions on $\partial\Omega_{T}$ to the homogeneous one by introducing in place of $u$ the new function $u-\phi$, \textcolor{black}{see Section 4 of Chapter V in \cite{Ladyzhenskaya} or Chapter 8 of \cite{GilbargTrudinger} for details on such considerations}. We do not explore this here, because our goal is not to prove the most general result possible, but to provide a concrete setup in which we can prove the validity of the approximation results of interest.}

The objective function is
\begin{align}
J(f)&=\norm{  \mathcal{G}[f]}^{2}_{2,\Omega_{T}} +   \norm{f}^{2}_{2, \partial \Omega_{T}} + \norm{ f(0, \cdot) - u_0 }^{2}_{2, \Omega}\nonumber
\end{align}


Recall that the norms above are $L^{2}(X)$ norms in the respective space $X=\Omega_{T}, \partial \Omega_{T}$ and $\Omega$ respectively. From Theorem \ref{T:UniversalApproximationThmPDE}, we have that
\begin{eqnarray}
J(f^{n}) \rightarrow 0 \text{ as }n\rightarrow \infty.\nonumber
\end{eqnarray}

Each neural network $f^{n}$ satisfies the PDE
\begin{align}
\mathcal{G}[f^{n}](t,x)&=h^{n}(t,x), \textrm{ for }(t,x)\in\Omega_{T} \nonumber\\
f^{n}(0,x)&=u^{n}_{0}(x), \textrm{ for }x\in\Omega\nonumber\\
f^{n}(t,x)&=g^{n}(t,x), \textrm{ for }(t,x)\in\partial\Omega_{T}\label{Eq:ApproxQuasiLinearParabolicPDE}
\end{align}
for some  $h^{n}, u^{n}_{0},$ and $g^{n} $  such that
\begin{align}
&\norm{  h^{n}}^{2}_{2,\Omega_{T}}+ \norm{  g^{n} }^{2}_{2, \partial \Omega_{T}} + \norm{ u^{n}_{0} - u_0 }^{2}_{2, \Omega}\rightarrow 0 \text{ as }n\rightarrow \infty.\label{Eq:ApproxQuasiLinearDecayRates}
\end{align}
For the purposes of this section, we  make the following set of assumptions.
\begin{condition}\label{A:MinimalAssumptions}
\begin{itemize}
\item{There is a constant  $\mu>0$ and positive functions $\kappa(t,x),\lambda(t,x)$ such that for all $(t,x)\in\Omega_{T}$ we have
\begin{align*}
\norm{\alpha(t,x,u,p)} \leq \mu(\kappa(t,x)+\norm{p}), \text{ and } |\gamma(t,x,u,p)|\leq \lambda(t,x)\norm{p},\nonumber
\end{align*}
with $\kappa \in L^{2}(\Omega_{T})$, $\lambda\in L^{d+2+\eta}(\Omega_{T})$ for some $\eta>0$.}
\item{$\alpha(t,x,u,p)$ and $\gamma(t,x,u,p)$ are Lipschitz continuous in $(t,x,u,p)\in \Omega_{T}\times\mathbb{R}\times\mathbb{R}^{d}$ uniformly on compacts of the form $\{(t,x)\in \bar{\Omega}_{T}, |u|\leq C, |p|\leq C\}$.}
\item{$\alpha(t,x,u,p)$ is differentiable with respect to $(x,u,p)$ with continuous derivatives. 
}
\item{There is a positive constant $\nu>0$ such that
\begin{align*}
\alpha(t,x,u,p)p\geq \nu |p|^{2}
\end{align*}
and
\begin{align*}
\left<\alpha(t,x,u,p_{1})-\alpha(t,x,u,p_{2}), p_{1}-p_{2}\right> >0, \text{ for every }  p_{1},p_{2}\in\mathbb{R}^{d}, p_{1}\neq p_{2}.\label{Eq:CondNonlinearCase}
\end{align*}
   }
\item{ $u_{0}(x)\in \mathcal{C}^{0,2+\xi}(\bar{\Omega})$ for some $\xi>0$\footnote{In general, the H\"{o}lder space $\mathcal{C}^{0,\xi}(\bar{\Omega})$ is the Banach space of continuous functions in $\bar{\Omega}$ having continuous derivatives up to order $[\xi]$ in $\bar{\Omega}$ with finite corresponding uniform norms and finite uniform $\xi-[\xi]$ H\"{o}lder norm. Analogously, we also define the H\"{o}lder space $\mathcal{C}^{0,\xi,\xi/2}(\bar{\Omega}_{T})$ which in addition has finite $[\xi]/2$ and $(\xi-[\xi])/2$ regular and H\"{o}lder derivatives norms in time respectively. These spaces are denoted by $H^{\xi}(\bar{\Omega})$ and $H^{\xi,\xi/2}(\bar{\Omega}_{T})$ respectively in \cite{Ladyzhenskaya}. } with itself and its first derivative bounded in $\bar{\Omega}$.}
\item{$\Omega$ is a bounded, open subset of $\mathbb{R}^{d}$ with boundary $\partial\Omega\in \mathcal{C}^{2}$.}
\item{For every $n\in\mathbb{N}$,  $f^{n}\in\mathcal{C}^{1,2}(\bar{\Omega}_{T})$. In addition, $(f^{n})_{n\in\mathbb{N}}\in L^{2}(\Omega_{T})$.}
\end{itemize}
\end{condition}

\begin{theorem}\label{T:MainConvTheorem}
Assume that Condition \ref{A:MinimalAssumptions} and (\ref{Eq:ApproxQuasiLinearDecayRates}) hold. Then, problem (\ref{Eq:QuasiLinearParabolicPDE2}) has a  unique bounded solution in 
$\mathcal{C}^{0,\delta,\delta/2}(\bar{\Omega}_{T})\cap L^{2}\left(0,T; W^{1,2}_{0}(\Omega)\right)\cap W^{(1,2),2}_{0}(\Omega_{T}^{\prime})$ for some $\delta>0$ and any interior subdomain $\Omega_{T}^{\prime}$ of $\Omega_{T}$\footnote{Here $W^{(1,2),2}_{0}(\Omega_{T}^{\prime})$ denotes the Banach space which is the closure of $C_{0}^{\infty}(\Omega_{T}^{\prime})$ with elements from $L^{2}(\Omega_{T}^{\prime})$ having generalized derivatives of the form $D^{r}_{t}D^{s}_{x}$ with $r,s$ such that $2r+s\leq 2$ with the usual Sobolev norm.}. 
 In addition,  $f^{n}$ converges  to $u$, the unique solution to (\ref{Eq:QuasiLinearParabolicPDE2}), strongly in $L^{\rho}(\Omega_{T})$ for every $\rho<2$. If, in addition, the sequence $\{f^{n}(t,x)\}_{n\in\mathbb{N}}$ is uniformly bounded in $n$ and equicontinuous then the convergence to $u$ is uniform in $\Omega_{T}$.
\end{theorem}

The proof of this theorem is in the Appendix. We conclude this section with some remarks and an example.
\begin{remark}
Despite the restriction made to the zero boundary data case, we do expect that our results are also valid for reasonably smooth inhomogeneous boundary data. In addition, if we make further assumptions on the nonlinearities $\alpha(t,x,u,p)$ and $\gamma(t,x,u,p)$ and on the initial data $u_{0}(x)$, then one can establish existence and uniqueness of classical solutions, see for example Section 6 of Chapter V in \cite{Ladyzhenskaya} for details.  As a matter of fact the results of Chapter V.6 in \cite{Ladyzhenskaya} show that with assuming a little bit more on the growth of the derivatives of the nonlinear functions $\alpha(t,x,u,p),\gamma(t,x,u,p)$ will lead to $\nabla_{x}u\in \mathcal{C}^{0,\delta',\delta'/2}(\Omega_{T})$ for some $\delta'>0$. Furthermore, we remark here that stronger claims can be made if more properties are known in regards to the given approximating family $\{f^{n}\}$ such as, for example,  a-priori bounds on appropriate Sobolev norms, but we do not explore this further here.
\end{remark}

\begin{remark}\label{R:EquicontinuousApprox}
The uniform, in $n$,  $L^{2}$ bound for the sequence $\{f^{n}\}_{n\in\mathbb{N}}$ is easily satisfied for a bounded neural network approximation sequence $f^{n}(t,x)$. However, we believe that it is true for a wider class of models, after all one expects that to be true if $f^{n}$ indeed converges in $L^{\rho}$ for $\rho<2$. The condition on equicontinuity for $\{f^{n}(t,x)\}$  allows to both simplify the proof and make a stronger claim as well. However, it is only a sufficient condition and not necessary. The paper, \cite{ChandraSingh2004}, see Theorems 19 and 20 therein, discusses structural restrictions (a-priori boundedness and summability) that can be imposed on the unknown weights of feedforward neural networks, belonging in the class $\mathfrak{C}(\psi)=\bigcup_{n\geq 1}\mathfrak{C}^{n}(\psi)$ as defined by (\ref{Eq:FeedForwardNN}), which then guarantee both equicontinuity and universal approximation properties of the neural network for continuous and bounded functions. As it is also discussed in \cite{ChandraSingh2004}, equicontinuity is also related to fault-tolerance properties of neural networks, a subject worthy of further study in the context of PDEs. However, we do not discuss this further here as this would be a topic for a different paper.
\end{remark}

Let us present the case of linear parabolic PDEs  in Example \ref{Ex:LinearCase} below.
\begin{example}[Linear case]\label{Ex:LinearCase}
Let us assume that the operator $\mathcal{G}$ is linear in $u$ and $\nabla u$. In particular, let us set
\[
\alpha_{i}(t,x,u,p)=\sum_{j=1}^{n}\left(\sigma\sigma^{T}\right)_{i,j}(t,x)p_{j}, i=1,\cdots d
\]
and
\[
\gamma(t,x,u,p)=-\left<b(t,x),p\right> +\sum_{i,j=1}^{d}\frac{\partial}{\partial x_{i}}\left(\sigma\sigma^{T}\right)_{i,j}(t,x)p_{j}-c(t,x) u.
\]

Assume that there are positive constants $\nu,\mu>0$ such that for every $\xi\in\mathbb{R}^{d}$ the matrix $\left[\left(\sigma\sigma^{T}\right)_{i,j}(t,x)\right]_{i,j=1}^{d} $ satisfies
\[
\nu |\xi|^{2}\leq \sum_{i,j=1}^{d}\left(\sigma\sigma^{T}\right)_{i,j}(t,x)\xi_{i}\xi_{j}\leq \mu |\xi|^{2}
\]
and that the coefficients $b$ and $c$ are such that
\begin{align}
&\norm{\sum_{i=1}^{d}b_{i}^{2}}_{q,r,\Omega_{T}}+\norm{c}_{q,r,\Omega_{T}}\leq \mu, \text{ for some } \mu>0\nonumber
\end{align}
where we recall for example $\norm{c}_{q,r,\Omega_{T}}=\left(\int_{0}^{T}\left(\int_{\Omega}|c(t,x)|^{q}dx\right)^{r/q}\right)^{1/r}$ and  $r,q$ satisfy the relations
\begin{align}
\frac{1}{r}+\frac{d}{2q}=1\nonumber\\
q\in(d/2,\infty], r\in[1,\infty), \textrm{ for } d\geq 2,\nonumber\\
q\in[1,\infty], r\in[1,2], \textrm{ for } d=1.\nonumber
\end{align}

In particular, the previous bounds always hold in the case of coefficients $b$ and $c$ that are bounded in $\Omega_{T}$. Under these conditions, standard results for linear PDE's, see for instance Theorem 4.5 of Chapter III of \cite{Ladyzhenskaya} for a related result, show that approximation results analogous to that of Theorem \ref{T:MainConvTheorem} hold. 
\end{example}

\section{Conclusion}  \label{Conclusion}
We believe that deep learning could become a valuable approach for solving high-dimensional PDEs, which are important in physics, engineering, and finance. The PDE solution can be approximated with a deep neural network which is trained to satisfy the differential operator, initial condition, and boundary conditions. We prove that the neural network converges to the solution of the partial differential equation as the number of hidden units increases.

Our deep learning algorithm for solving PDEs is meshfree, which is key since meshes become infeasible in higher dimensions. Instead of forming a mesh, the neural network is trained on batches of randomly sampled time and space points. The approach is implemented for a class of high-dimensional free boundary PDEs in up to $200$ dimensions with accurate results. We also test it on a high-dimensional Hamilton-Jacobi-Bellman PDE with accurate results.

\textcolor{black}{The DGM algorithm can be easily modified to apply to hyperbolic, elliptic, and partial-integral differential equations. The algorithm remains essentially the same for these other types of PDEs. However, numerical performance for these other types of PDEs remains to be be investigated.}

It is also important to put the numerical results in Sections \ref{NumericalAnalysis}, \ref{HJBsection} and \ref{BurgerEquation} in a proper context. PDEs with highly non-monotonic or oscillatory solutions may be more challenging to solve and further developments in architecture will be necessary. Further numerical development and testing is therefore required to better judge the usefulness of deep learning for the solution of PDEs in other applications. However, the numerical results of this paper demonstrate that there is sufficient evidence to further explore deep neural network approaches for solving PDEs.


In addition, it would be of interest to establish results analogous to Theorem \ref{T:MainConvTheorem} for PDEs beyond the class of \textcolor{black}{quasilinear parabolic PDEs} considered in this paper. \textcolor{black}{Stability analysis of deep learning and machine learning algorithms for solving PDEs is also an important question. It would certainly be interesting to study machine learning algorithms that use a more direct variational formulation of the involved PDEs.} We leave these questions for future work.

\appendix

\section{Proofs of convergence results}\label{S:Proofs}
In this section we have gathered the proofs of the theoretical results of Section  \ref{ApproximationProof}.

\begin{proof}[Proof of Theorem \ref{T:UniversalApproximationThmPDE}]

By Theorem 3 of \cite{Hornik} we know that there is a function $f\in\mathfrak{C}(\psi) $ that is uniformly $2-$dense on compacts of $\mathcal{C}^{2}(\mathbb{R}^{1+d})$.  This means that for $u\in\mathcal{C}^{1,2}([0,T]\times\mathbb{R}^{d})$ and $\epsilon>0$, there is $f\in\mathfrak{C}(\psi) $ such that
\begin{align}
\sup_{(t,x)\in\Omega_{T}}|\partial_{t}u(t,x)-\partial_{t} f(t,x;\theta)|+\max_{|a|\leq 2}\sup_{(t,x)\in \bar{\Omega}_{T}} |\partial^{(a)}_{x}u(t,x)-\partial^{(a)}_{x}f(t,x;\theta)|<\epsilon\label{Eq: HornikResult}
\end{align}

\textcolor{black}{We have assumed that $(u,p)\mapsto\hat{\gamma}(t,x,u,p)$  is locally Lipschitz continuous in $(u,p)$ with Lipschitz constant that can have at most polynomial growth in $u$ and $p$ , uniformly with respect to $t,x$. This means that
\begin{align}
\left|\hat{\gamma}(t,x,u,p)-\hat{\gamma}(t,x,v,s)\right|&\leq \left(|u|^{q_{1}/2}+|p|^{q_{2}/2}+|v|^{q_{3}/2}+|s|^{q_{4}/2}\right)\left(|u-v|+|p-s|\right).\nonumber
\end{align}
for some constants $0\leq q_{1},q_{2},q_{3},q_{4}<\infty$. Therefore we obtain, using H\"{o}lder inequality with exponents $r_{1},r_{2}$,
\begin{align}
&\int_{\Omega_{T}}\left|\hat{\gamma}(t,x,f,\nabla_{x}f)-\hat{\gamma}(t,x,u,\nabla_{x}u)\right|^{2}d\nu_{1}(t,x)\leq\nonumber\\
&\leq \int_{\Omega_{T}}\left(|f(t,x;\theta)|^{q_{1}}+|\nabla_{x}f(t,x;\theta)|^{q_{2}}+|u(t,x)|^{q_{3}}+|\nabla_{x}u(t,x)|^{q_{4}}\right) \nonumber\\
&\qquad\times\left(|f(t,x;\theta)-u(t,x)|^{2}+|\nabla_{x}f(t,x;\theta)-\nabla_{x}u(t,x)|^{2}\right)d\nu_{1}(t,x)\nonumber\\
&\leq \left(\int_{\Omega_{T}}\left(|f(t,x;\theta)|^{q_{1}}+|\nabla_{x}f(t,x;\theta)|^{q_{2}}+|u(t,x)|^{q_{3}}+|\nabla_{x}u(t,x)|^{q_{4}}\right)^{r_{1}}d\nu_{1}(t,x)\right)^{1/r_{1}}\nonumber\\
&\qquad \times\left(\int_{\Omega_{T}}\left(|f(t,x;\theta)-u(t,x)|^{2}+|\nabla_{x}f(t,x;\theta)-\nabla_{x}u(t,x)|^{2}\right)^{r_{2}}d\nu_{1}(t,x)\right)^{1/r_{2}}\nonumber\\
&\leq K\left(\int_{\Omega_{T}}\left(|f(t,x;\theta)-u(t,x)|^{q_{1}}+|\nabla_{x}f(t,x;\theta)-\nabla_{x}u(t,x)|^{q_{2}}+|u(t,x)|^{q_{1}\vee q_{3}}+|\nabla_{x}u(t,x)|^{q_{2}\vee q_{4}}\right)^{r_{1}}d\nu_{1}(t,x)\right)^{1/r_{1}}\nonumber\\
&\qquad \times\left(\int_{\Omega_{T}}\left(|f(t,x;\theta)-u(t,x)|^{2}+|\nabla_{x}f(t,x;\theta)-\nabla_{x}u(t,x)|^{2}\right)^{r_{2}}d\nu_{1}(t,x)\right)^{1/r_{2}}\nonumber\\
& \leq K\left(\epsilon^{q_{1}}+\epsilon^{q_{2}}+\sup_{\Omega_{T}}|u|^{q_{1}\vee q_{3}}+\sup_{\Omega_{T}}|\nabla_{x}u|^{q_{2}\vee q_{4}}\right)\epsilon^{2}\label{Eq:BoundNonlinearTerm}
\end{align}
where the unimportant constant $K<\infty$ may change from line to line and for two numbers $q_{1}\vee q_{3}=\max\{q_{1},q_{3}\}$.} In the last step we used (\ref{Eq: HornikResult}).

\textcolor{black}{In addition, we have also assumed that for every $i,j\in\{1,\cdots d\}$, the mapping $(u,p)\mapsto \frac{\partial\alpha_{i}(t,x,u, p)}{\partial p_{j}}$ is locally Lipschitz in $(u,p)$ with Lipschitz constant that can have at most polynomial growth on $u$ and $p$, uniformly with respect to $t,x$.  This means that
\begin{align}
\left|\frac{\partial\alpha_{i}(t,x,u, p)}{\partial p_{j}}-\frac{\partial\alpha_{i}(t,x,v, s)}{\partial s_{j}}\right|&\leq \left(|u|^{q_{1}/2}+|p|^{q_{2}/2}+|v|^{q_{3}/2}+|s|^{q_{4}/2}\right)\left(|u-v|+|p-s|\right).\nonumber
\end{align}
for some constants $0\leq q_{1},q_{2},q_{3},q_{4}<\infty$. Denote for convenience
\[
\xi(t,x,u,\nabla u, \nabla^{2}u)=\sum_{i,j=1}^{d}\frac{\partial\alpha_{i}(t,x,u(t,x), \nabla u(t,x))}{\partial u_{x_{j}}}\partial_{x_{i},x_{j}}u(t,x).
\]
Then, similarly to (\ref{Eq:BoundNonlinearTerm}) we have after an application of H\"{o}lder inequality, for some constant $K<\infty$ that may change from line to line,
\begin{align}
&\int_{\Omega_{T}}\left|\xi(t,x,f,\nabla_{x}f, \nabla^{2}_{x}f)-\xi(t,x,u,\nabla_{x}u,\nabla^{2}_{x}u)\right|^{2}d\nu_{1}(t,x)\leq\nonumber\\
&\leq\int_{\Omega_{T}}\left|\sum_{i,j=1}^{d}\left(\frac{\partial\alpha_{i}(t,x,f(t,x;\theta), \nabla f(t,x;\theta))}{\partial f_{x_{j}}}-\frac{\partial\alpha_{i}(t,x,u(t,x), \nabla u(t,x))}{\partial u_{x_{j}}}\right)\partial_{x_{i},x_{j}}u(t,x)\right|^{2}d\nu_{1}(t,x)\nonumber\\
&\quad+\int_{\Omega_{T}}\left|\sum_{i,j=1}^{d}\frac{\partial\alpha_{i}(t,x,f(t,x;\theta), \nabla f(t,x;\theta))}{\partial f_{x_{j}}}\left(\partial_{x_{i},x_{j}}f(t,x;\theta)-\partial_{x_{i},x_{j}}u(t,x)\right)\right|^{2}d\nu_{1}(t,x)\nonumber\\
&\leq K \sum_{i,j=1}^{d}\left(\int_{\Omega_{T}}\left|\partial_{x_{i},x_{j}}u(t,x)\right|^{2p}d\nu_{1}(t,x)\right)^{1/p}\times\nonumber\\
&\quad\times\left(
\int_{\Omega_{T}}\left|\frac{\partial\alpha_{i}(t,x,f(t,x;\theta), \nabla f(t,x;\theta))}{\partial f_{x_{j}}}-\frac{\partial\alpha_{i}(t,x,u(t,x), \nabla u(t,x))}{\partial u_{x_{j}}}\right|^{2q}d\nu_{1}(t,x)\right)^{1/q}+\nonumber\\
&\quad+ K \sum_{i,j=1}^{d}\left(\int_{\Omega_{T}}\left|\frac{\partial\alpha_{i}(t,x,f, \nabla f)}{\partial f_{x_{j}}}\right|^{2p} d\nu_{1}(t,x)\right)^{1/p} \left(\int_{\Omega_{T}}\left|\partial_{x_{i},x_{j}}f(t,x;\theta)-\partial_{x_{i},x_{j}}u(t,x)\right|^{2q}d\nu_{1}(t,x)\right)^{1/q}\nonumber
\end{align}
\begin{align}
&\leq K \sum_{i,j=1}^{d}\left(\int_{\Omega_{T}}\left|\partial_{x_{i},x_{j}}u(t,x)\right|^{2p}d\nu_{1}(t,x)\right)^{1/p} \times\nonumber\\
&\quad\times\left(\int_{\Omega_{T}}\left(|f(t,x;\theta)-u(t,x)|^{q_{1}}+|\nabla_{x}f(t,x;\theta)-\nabla_{x}u(t,x)|^{q_{2}}+|u(t,x)|^{q_{1}\vee q_{3}}+|\nabla_{x}u(t,x)|^{q_{2}\vee q_{4}}\right)^{qr_{1}}d\nu_{1}(t,x)\right)^{1/(qr_{1})}\nonumber\\
&\quad \times\left(\int_{\Omega_{T}}\left(|f(t,x;\theta)-u(t,x)|^{2}+|\nabla_{x}f(t,x;\theta)-\nabla_{x}u(t,x)|^{2}\right)^{qr_{2}}d\nu_{1}(t,x)\right)^{1/(qr_{2})}\nonumber\\
&\quad+ K \sum_{i,j=1}^{d}\left(\int_{\Omega_{T}}\left|\frac{\partial\alpha_{i}(t,x,f, \nabla f)}{\partial f_{x_{j}}}\right|^{2p} d\nu_{1}(t,x)\right)^{1/p} \left(\int_{\Omega_{T}}\left|\partial_{x_{i},x_{j}}f(t,x;\theta)-\partial_{x_{i},x_{j}}u(t,x)\right|^{2q}d\nu_{1}(t,x)\right)^{1/q}\nonumber\\
& \leq K\epsilon^{2},\label{Eq:BoundNonlinearTerm2}
\end{align}
where in the last step we followed the computation in (\ref{Eq:BoundNonlinearTerm}) and used  (\ref{Eq: HornikResult}).}

Using (\ref{Eq: HornikResult}) and (\ref{Eq:BoundNonlinearTerm})-(\ref{Eq:BoundNonlinearTerm2}) we subsequently obtain for the objective function \textcolor{black}{(note that $\mathcal{G}[u](t,x)=0$ for $u$ that solves the PDE)
 \begin{align}
 J(f)&=\norm{  \mathcal{G}[f](t,x)}_{\Omega_{T}, \nu_1}^2 +   \norm{  f(t,x; \theta ) - g(t,x) }_{ \partial \Omega_{T}, \nu_2}^2 + \norm{ f(0, x; \theta) - u_0(x) }_{ \Omega, \nu_3}^2\nonumber\\
 &=\norm{ \mathcal{G}[f](t,x)-\mathcal{G}[u](t,x) }_{\Omega_{T}, \nu_1}^2 +   \norm{  f(t,x; \theta ) - g(t,x) }_{ \partial \Omega_{T}, \nu_2}^2 + \norm{ f(0, x; \theta) - u_0(x) }_{ \Omega, \nu_3}^2\nonumber\\
 &\leq \int_{\Omega_{T}}\left|\partial_{t}u(t,x)-\partial_{t} f(t,x;\theta)\right|^{2}d\nu_{1}(t,x)+
\int_{\Omega_{T}}\left|\xi(t,x,f,\nabla f, \nabla^{2}f)-\xi(t,x,u,\nabla u, \nabla^{2}u)\right|^{2}d\nu_{1}(t,x)\nonumber\\
 &+\int_{\Omega_{T}}\left|\gamma(t,x,f,\nabla_{x}f)-\gamma(t,x,u,\nabla_{x}u)\right|^{2}d\nu_{1}(t,x)+
 \int_{\partial\Omega_{T}}\left|f(t,x;\theta)-u(t,x)\right|^{2}d\nu_{2}(t,x)
 +\nonumber\\
 &\qquad+\int_{\Omega}\left|f(0,x;\theta)-u(0,x)\right|^{2}d\nu_{3}(t,x)\nonumber\\
 & \leq K\epsilon^{2}\nonumber
 \end{align}}
for an appropriate constant $K<\infty$. The last step completes the proof of the Theorem after rescaling $\epsilon$.
\end{proof}

\begin{proof}[Proof of Theorem \ref{T:MainConvTheorem}]

\textcolor{black}{Existence, regularity and uniqueness for (\ref{Eq:QuasiLinearParabolicPDE2}) follows from Theorem 2.1 \cite{Porzio}  combined with Theorems 6.3-6.5 of Chapter V.6 in \cite{Ladyzhenskaya} (see also Theorem 6.6 of Chapter V.6 of \cite{Ladyzhenskaya}). Boundedness follows from Theorem 2.1 in \cite{Porzio} and Chapter V.2 in \cite{Ladyzhenskaya}. The convergence proof follows by the smoothness of the neural networks together with compactness arguments as we explain below.}

\textcolor{black}{Let us first consider problem (\ref{Eq:ApproxQuasiLinearParabolicPDE}) with $g^{n}(t,x)=0$ and let us denote the solution to this problem by $\hat{f}^{n}(t,x)$.} 
\textcolor{black}{Due to Condition \ref{A:MinimalAssumptions}, Lemma 4.1 of \cite{Porzio}  applies and gives that $\{\hat{f}^{n}\}_{n\in\mathbb{N}}$ is uniformly bounded with respect to $n$ in at least $L^{\infty}\left(0,T;L^{2}(\Omega)\right)\cap L^{2}\left(0,T; W^{1,2}_{0}(\Omega)\right)$ (in regards to such uniform energy bound results we also refer the reader to Theorem 2.1 and Remark 2.14 of \cite{BoccardoPorzioPrimo} for the case $\gamma=0$ and to \cite{Magliocca, NardoEtAll}  for related results in more general cases). As a  matter of fact $\hat{f}^{n}$ is more regular than stated, see Section 6, Chapter V of \cite{Ladyzhenskaya}, but we will not make use of this fact in the convergence proof of $\hat{f}^{n}$ to $u$. These uniform energy bounds imply that we can extract a subsequence, denoted also by $\{\hat{f}^{n}\}_{n\in\mathbb{N}}$, which converges to some $u$ in the weak-* sense in $L^{\infty}\left(0,T;L^{2}(\Omega)\right)$ and weakly in $L^{2}\left(0,T; W^{1,2}_{0}(\Omega)\right)$ and to some $v$ weakly in $L^{2}(\Omega)$ for every fixed $t\in(0,T]$. }

\textcolor{black}{Next let us set $q=1+\frac{d}{d+4}\in (1,2)$ and note that for conjugates, $r_{1},r_{2}>1$ such that $1/r_{1}+1/r_{2}=1$
\begin{align}
\int_{\Omega_{T}}\left|\gamma(t,x,\hat{f}^{n},\nabla_{x}\hat{f}^{n})\right|^{q}dtdx &\leq \int_{\Omega_{T}}\left|\lambda(t,x)\right|^{q}|\nabla_{x}\hat{f}^{n}(t,x)|^{q}dtdx\nonumber\\
&\leq\left(\int_{\Omega_{T}}\left|\lambda(t,x)\right|^{r_{1}q}dtdx\right)^{1/r_{1}}\left(\int_{\Omega_{T}}|\nabla_{x}\hat{f}^{n}(t,x)|^{r_{2}q} dtdx\right)^{1/r_{2}}.\label{Eq:MomentBoundGammaTerm}
\end{align}
}

\textcolor{black}{Let us choose $r_{2}=2/q>1$. Then we calculate $r_{1}=\frac{r_{2}}{r_{2}-1}=\frac{2}{2-q}$. Hence, we have that $r_{1}q=d+2$. Recalling the assumption $\lambda\in L^{d+2}(\Omega_{T})$ and the uniform bound on the $\norm{\nabla_{x}\hat{f}^{n}}_{2}$ we subsequently obtain that for $q=1+\frac{d}{d+4}$, there is a constant $C<\infty$ such that
\begin{align}
\int_{\Omega_{T}}\left|\gamma(t,x,\hat{f}^{n},\nabla_{x}\hat{f}^{n})\right|^{q}dtdx &\leq C.\nonumber
\end{align}}

\textcolor{black}{The latter estimate together with the  growth assumptions on $\alpha(\cdot)$  from Condition \ref{A:MinimalAssumptions}, imply that $\{\partial_{t}\hat{f}^{n}\}_{n\in\mathbb{N}}$ is bounded uniformly with respect to $n$ in $L^{1+d/(d+4)}(\Omega_{T})$ and in $L^{2}(0,T; W^{-1,2}(\Omega))$. Consider the conjugates $1/\delta_{1}+1/\delta_{2}=1$  with $\delta_{2}> \max\{2,d\}$. Due to the embedding
\[
W^{-1,2}(\Omega)\subset W^{-1,\delta_{1}}(\Omega),\quad L^{q}(\Omega)\subset W^{-1,\delta_{1}}(\Omega), \text{ and } L^{2}(\Omega)\subset W^{-1,\delta_{1}}(\Omega),
\]
we have that $\{\partial_{t}\hat{f}^{n}\}_{n\in\mathbb{N}}$ is bounded uniformly with respect to $n$ in $L^{1}(0,T; W^{-1,\delta_{1}}(\Omega))$. Define now the spaces $X=W^{1,2}_{0}(\Omega)$, $B=L^{2}(\Omega)$ and $Y=W^{-1,\delta_{1}}(\Omega)$, and notice that
\[
X\subset B\subset Y
\]
with the first embedding being compact.  Then, Corollary 4 of \cite{Simon} yields relative compactness of $\{\hat{f}^{n}\}_{n\in\mathbb{N}}$ in $L^{2}(\Omega_{T})$, which means that $\{\hat{f}^{n}\}_{n\in\mathbb{N}}$ converges strongly to $u$ in that space. Thus, up to subsequences, $\{\hat{f}^{n}\}_{n\in\mathbb{N}}$ converges almost everywhere to u in $\Omega_{T}$.}

\textcolor{black}{The nonlinearity of the $\alpha$ and $\gamma$ functions with respect to the gradient prohibits us from passing to the limit directly in the respective weak formulation. However, the uniform boundedness of $\{\hat{f}_{n}\}_{n\in\mathbb{N}}$ in $L^{\sigma}\left(0,T; W^{1,\sigma}_{0}(\Omega)\right)$ with $\sigma>1$ (in fact here $\sigma=2$) and its weak convergence to $u$ in that space,  allows us to conclude, as in Theorem 3.3 of \cite{BDGO1997},  that
\[
\nabla \hat{f}^{n}\rightarrow \nabla u \textrm{ almost everywhere in } \Omega_{T}.
\]}

\textcolor{black}{Hence, we obtain that $\{\hat{f}^{n}\}_{n\in\mathbb{N}}$ converges  to $u$ strongly also in $L^{\rho}\left(0,T; W^{1,\rho}_{0}(\Omega)\right)$ for every $\rho<2$. }

\textcolor{black}{In preparation to passing to the limit as $n\rightarrow \infty$ in the weak formulation, we need to study the behavior of the nonlinear terms. Recalling the assumptions on $\alpha(t,x,u,p)$ we have for $\rho<2$ and for a measurable set $A\subset\Omega_{T}$ (the constant $K<\infty$ may change from line to line)
\begin{align}
&\int_{A}\left|\alpha(t,x,\hat{f}^{n},\nabla \hat{f}^{n})\right|^{\rho}dtdx\leq K \left[\int_{A}|\kappa(t,x)|^{\rho}dtdx+\int_{A}|\nabla \hat{f}^{n}(t,x)|^{\rho}dtdx\right]\nonumber\\
&\qquad\leq K \left[\int_{A}|\kappa(t,x)|^{\rho}dtdx+\left(\int_{\Omega_{T}}|\nabla \hat{f}^{n}(t,x)|^{2}dtdx\right)^{\rho/2}|A|^{1-\rho/2}\right]\nonumber\\
&\qquad\leq K \left[\int_{A}|\kappa(t,x)|^{\rho}dtdx+|A|^{1-\rho/2}\right].\nonumber
\end{align}
In the latter display we used H\"{o}der inequality with exponent $2/\rho>1$.} \textcolor{black}{By Vitali's theorem we then conclude that
\[
\alpha(t,x,\hat{f}^{n},\nabla \hat{f}^{n})\rightarrow \alpha(t,x,u,\nabla u) \textrm{ strongly in } L^{\rho}(\Omega_{T})
\]
as $n\rightarrow\infty$, for every $1<\rho<2$. For the same reason, an analogous estimate to (\ref{Eq:MomentBoundGammaTerm}), gives
\[
\int_{A}\left|\gamma(t,x,\hat{f}^{n},\nabla_{x}\hat{f}^{n})\right|^{q}dtdx \leq  K \left(\int_{A}\left|\lambda(t,x)\right|^{d+2}dtdx\right)^{(2-q)/2} \leq K |A|^{\frac{\eta}{d+2+\eta}}
\]
 implying, via Vitali's theorem,  that
\[
\gamma(t,x,\hat{f}^{n},\nabla \hat{f}^{n})\rightarrow \gamma(t,x,u,\nabla u) \textrm{ strongly in } L^{q}(\Omega_{T})
\]
as $n\rightarrow\infty$, for $q=1+\frac{d}{d+4}$.}

\textcolor{black}{
Notice also that by construction we have that the initial condition $u^{n}_{0}$  converges to $u_{0}$ strongly in $L^{2}(\Omega)$. The weak formulation of the PDE (\ref{Eq:ApproxQuasiLinearParabolicPDE}) with $g^{n}=0$ reads as follows. For every $t_{1}\in(0,T]$
\begin{align}
&\int_{\Omega_{t_{1}}}\left[-\hat{f}^{n}\partial_{t}\phi +\left<\alpha(t,x,\hat{f}^{n},\nabla \hat{f}^{n}), \nabla \phi \right>+(\gamma(t,x,\hat{f}^{n},\nabla \hat{f}^{n})-h^{n})\phi\right](t,x)dxdt\nonumber\\
&\quad +\int_{\Omega}\hat{f}^{n}(t_{1},x)\phi(t_{1},x)dx-\int_{\Omega}u^{n}_{0}(x)\phi(0,x)dx=0\nonumber
\end{align}
for every $\phi\in C^{\infty}_{0}(\Omega_{T})$. Using the above convergence results, we then obtain that the limit point $u$ satisfies for every $t_{1}\in(0,T]$ the equation
\begin{align}
&\int_{\Omega_{t_{1}}}\left[-u\partial_{t}\phi +\left<\alpha(t,x,u,\nabla u), \nabla \phi \right>+\gamma(t,x,u,\nabla u)\phi\right](t,x)dxdt\nonumber\\
&\quad +\int_{\Omega}u(t_{1},x)\phi(t_{1},x)dx-\int_{\Omega}u_{0}(x)\phi(0,x)dx=0,\nonumber
\end{align}
which is the weak formulation of the equation (\ref{Eq:QuasiLinearParabolicPDE2}).}

\textcolor{black}{It remains to discuss the convergence of $f^{n}-\hat{f}^{n}$ to zero, where we recall that $f^{n}$ is the neural network approximation satisfying (\ref{Eq:ApproxQuasiLinearParabolicPDE}) and $\hat{f}^{n}$ satisfies (\ref{Eq:ApproxQuasiLinearParabolicPDE}) with $g^{n}=0$. The functions $f^{n}\in\mathcal{C}^{1,2}(\bar{\Omega}_{T})$ and $\bar{\Omega}_{T}$ is compact. We have also assumed that $\{f^{n}\}_{n}$ is uniformly bounded in $L^{2}(\Omega_{T})$. This implies that, up to a subsequence, $f^{n}$ will converge at least weakly in $L^{2}(\Omega_{T})$. Moreover,  the boundary values $g^{n}(t,x)$ (which is nothing else by $f^{n}(t,x)$ evaluated at the smooth boundary $\partial\Omega_{T}$) in $(\ref{Eq:ApproxQuasiLinearParabolicPDE})$  converge to zero strongly in $L^{2}$. It is then a standard result that  $g^{n}$, i.e., $f^{n}$ evaluated at the boundary, converges to zero, at least, almost uniformly along a subsequence, see for example Lemma 2.1 in Chapter II of \cite{Ladyzhenskaya}. As it then follows, for example, by the proof of Theorems 6.3-6.4-6.5  of Chapter V.6 in \cite{Ladyzhenskaya}, using smoothness and uniqueness,  $f^{n}$ will differ from the solution to the PDE $(\ref{Eq:ApproxQuasiLinearParabolicPDE})$ with $g^{n}=0$,  $\hat{f}^{n}(t,x)$, by a negligible amount as $n\rightarrow\infty$ in the almost everywhere sense. The assumed uniform $L^{2}$ bound for $\{f^{n}\}_{n\in\mathbb{N}}$ together with the previously derived uniform $L^{2}(\Omega_{T})$ bound for $\{\hat{f}^{n}\}_{n\in\mathbb{N}}$ yield uniform $L^{2}(\Omega_{T})$ boundedness for $\{f^{n}-\hat{f}^{n}\}_{n\in\mathbb{N}}$. Then, by Vitali's theorem again, we get that $\{f^{n}-\hat{f}^{n}\}_{n\in\mathbb{N}}$ goes to zero strongly in $L^{\rho}(\Omega_{T})$ for every $\rho<2$. }

\textcolor{black}{The previously derived strong convergence of $\{f^{n}-\hat{f}^{n}\}_{n\in\mathbb{N}}$  to zero  in $L^{\rho}(\Omega_{T})$ for every $\rho<2$, together with the strong $L^{2}(\Omega_{T})$ convergence of $\{\hat{f}^{n}\}_{n\in\mathbb{N}}$  to $u$, conclude the proof of the convergence in $L^{\rho}(\Omega_{T})$ for every $\rho<2$ using triangle inequality.}

\textcolor{black}{If $\{f^{n}(t,x)\}$ is equicontinuous then, Lemma 3.2 of \cite{Garcia1970} gives uniform convergence of $g^{n}$ to zero. Hence, by the previous analysis, it will certainly be true that  $\{f^{n}\}_{n\in\mathbb{N}}$  converges to $u$ in $L^{\rho}(\Omega_{T})$ for every $\rho<2$. The $L^{\rho}$ convergence to zero together with boundedness and equicontinuity of the sequence $\{f^{n}(t,x)\}$ results then in uniform convergence due to the well known Arzel\`{a}-Ascoli theorem. }
\end{proof}


\begin{thebibliography}{99}



\bibitem{Glynn} S. Asmussen and P. Glynn, Stochastic Simulation: Algorithms and Analysis, Springer, 2007.


\bibitem{Jentzen}  C. Beck, W. E., A. Jentzen.  Machine learning approximation algorithms for high-dimensional fully nonlinear partial differential equations and second-order backward stochastic differential equations, arXiv:1709.05963, 2017.




\bibitem{BertsekasThitsiklis2000} D. Bertsekas and J. Tsitsiklis, Gradient convergence in gradient methods via errors, \textit{SIAM Journal of Optimization}, Vol.10, No. 3, pgs. 627-642, 2000.


\bibitem{BDGO1997} L. Boccardo, A. Dall`Aglio, T. Gallou\"{e}t and L. Orsina, Nonlinear parabolic equations with measure data, \textit{Journal of Functional Analysis}, Vol. 147, pp. 237-258, (1997).

\bibitem{BoccardoPorzioPrimo} L. Boccardo, M.M. Porzio and A. Primo, Summability and existence results for nonlinear parabolic equations, \textit{Nonlinear Analysis: Theory, Methods and Applications}, Vol. 71, Issue 304, pp. 1-15, 2009.

\bibitem{Bungartz2} H. Bungartz, A. Heinecke, D. Pfluger, and S. Schraufstetter, Option pricing with a direct adaptive sparse grid approach, \textit{Journal of Computational and Applied Mathematics}, Vol. 236, No. 15, pgs. 3741-3750, 2012.

\bibitem{Griebel} H. Bungartz and M Griebel, Sparse Grids, \textit{Acta numerica}, Vol. 13, pgs. 174-269, 2004.

\bibitem{ChandraSingh2004}
P. Chandra and Y. Singh, Feedforward sigmoidal networks - equicontinuity and fault-tolerance properties,  \textit{IEEE Transactions on Neural Networks},   Vol. 15, Issue 6,  pp. 1350-1366, 2004.

\bibitem{Chaudhari}  P. Chaudhari, A. Oberman, S. Osher, S. Soatto, and G. Carlier.  Deep relaxation: partial differential equations for optimizing deep neural networks, 2017.

\bibitem{Cerrai} S. Cerrai, Stationary Hamilton-Jacobi Equations in Hilbert Spaces and Applications to a Stochastic Optimal Control Problem, \textit{SIAM Journal on Control and Optimization}, Vol. 40, No. 3, pgs. 824-852, 2001.

\bibitem{Cybenko} G. Cybenko, Approximation by superposition of a sigmoidal function, \textit{Mathematics of Control, Signals and Systems}, Vol. 2, pp. 303-314, 1989.


\bibitem{Gaines}  A. Davie and J. Gaines, Convergence of Numerical Schemes for the Solution of Parabolic Stochastic Partial Differential Equations, \textit{Mathematics of Computation}, Vol. 70, No. 233, pgs.  121-134, 2000.


\bibitem{Dean} J. Dean, G. Corrado, R. Monga, K. Chen, M. Devin, M. Mao, A. Senior, P. Tucker, K. Yang, Q. Le, and A. Ng. Large scale distributed deep networks. In Advances in Neural Information Processing Systems, pp. 1223-1231, 2012.



\bibitem{Debussche} A. Debussche, M. Fuhrman, and G. Tessitore, Optimal control of a stochastic heat equation with boundary-noise and boundary-control, \textit{ESAIM: Control, Optimisation and Calculus of Variations}, Vol. 13, No. 1, pgs. 178-205, 2007.




\bibitem{Weinan}  W. E., J. Han, and A. Jentzen.  Deep learning-based numerical methods for high-dimensional parabolic partial differential equations and backward stochastic differential equations, \textit{Communications in Mathematics and Statistics}, Springer, 2017.


\bibitem{Masaaki}  M. Fujii, A. Takahashi, and M. Takahashi.  Asymptotic Expansion as Prior Knowledge in Deep Learning Method for high dimensional BSDEs, arXiv:1710.07030, 2017.

\bibitem{Garcia1970} A. M. Garcia, E. Rodemich, H. Rumsey Jr. and M. Rosenblatt,  A real variable lemma and the continuity of paths of some Gaussian processes, \textit{Indiana University Mathematics Journal},
Vol. 20, No. 6, pp. 565-578, December, 1970.

\bibitem{Xavier}  X. Glorot and Y. Bengio.  Understanding the difficulty of training deep feedforward neural networks. In Proceedings of the Thirteenth International Conference on Artificial Intelligence and Statistics, pp. 249-256, 2010.

\bibitem{GainesLectureNotes} J. Gaines. Numerical experiments with SPDEs. London Mathematical Society Lecture Note Series, 55-71, 1995.


\bibitem{GilbargTrudinger}
D. Gilbarg and N.S. Trudinger. Elliptic partial differential equations of second order, second edition. \textit{Springer-Verlang, Berlin Heidelberg}, 1983.

\bibitem{Gyongy} I. Gy\"{o}ngy,  Lattice Approximations for Stochastic Quasi-Linear Parabolic Partial Differential Equations Driven by Space-Time White Noise I, \textit{Potential Analysis}, Vol. 9, Issue 1, pgs. 1–25, 1998.



\bibitem{Bungartz1} A. Heinecke, S. Schraufstetter, and H. Bungartz, A highly parallel Black-Scholes solver based on adaptive sparse grids, \textit{International Journal of Computer Mathematics}, Vol. 89, No. 9, pgs. 1212-1238, 2012.



\bibitem{Haugh} M. Haugh and L. Kogan, Pricing American Options: A Duality Approach, \textit{Operations Research}, Vol. 52, No. 2, pgs. 258-270, 2004.

\bibitem{SchmidhuberLSTM} S. Hochreiter and J. Schmidhuber, Long short-term memory, \textit{Neural Computation}, Vol. 9, No. 8, pgs. 1735-1780, 1997.



\bibitem{HornikEtAl} K. Hornik, M. Stinchcombe and H. White, Universal approximation of an unknown mapping and its derivatives using multilayer feedforward networks, \textit{Neural Networks}, Vol. 3, Issue 5, pp. 551-560, 1990.

\bibitem{Hornik} K. Hornik, Approximation capabilities of multilayer feedforward networks, \textit{Neural Networks}, Vol. 4, pp. 251-257, 1991.

\bibitem{ADAM} D. Kingma and J. Ba, ADAM: A method for stochastic optimization. arXiv: 1412.6980, 2014.

\bibitem{Ladyzhenskaya}  O. A. Ladyzenskaja, V. A. Solonnikov and N. N. Ural'ceva, Linear and Quasi-linear Equations of Parabolic Type (Translations of Mathematical Monographs Reprint), \textit{American Mathematical Society}, Vol. 23, 1988.

\bibitem{Lagaris}  I. Lagaris, A. Likas, and D. Fotiadis, Artificial neural networks for solving ordinary and partial differential equations, \textit{IEEE Transactions on Neural Networks}, Vol. 9, No. 5, pgs. 987-1000, 1998.

\bibitem{Likas} I. Lagaris, A. Likas, and D. Papageorgiou, Neural-network methods for boundary value problems with irregular boundaries, \textit{IEEE Transactions on Neural Networks}, Vol. 11, No. 5, pgs. 1041-1049, 2000.

\bibitem{Lee}  H. Lee, Neural Algorithm for Solving Differential Equations, \textit{Journal of Computational Physics}, Vol. 91, pgs. 110-131, 1990.

\bibitem{Ling} J. Ling, A. Kurzawski, and J. Templeton. Reynolds averaged turbulence modelling using deep neural networks with embedded invariance. \textit{Journal of Fluid Mechanics}, Vol. 807, pgs. 155-166, 2016.


\bibitem{LongstaffSchwartz} F. Longstaff and E. Schwartz, Valuing American Options by Simulation: A Simple Least-Squares Approach, \textit{Review of Financial Studies}, Vol. 14, pgs.  113-147, 2001.


\bibitem{Magliocca} M. Magliocca,  Existence results for a Cauchy-Dirichlet parabolic problem with a repulsive gradient term, \textit{Nonlinear Analysis}, Vol. 166, pp. 102-143, 2018


\bibitem{Malek}  A. Malek and R. Beidokhti, Numerical solution for high order differential equations using a hybrid neural network-optimization method, \textit{Applied Mathematics and Computation}, Vol. 183, No. 1, pgs.  260-271, 2006.


\bibitem{Masiero} F. Masiero, HJB equations in infinite dimensions, \textit{Journal of Evolution Equations}, Vol. 16, No. 4, pgs. 789-824, 2016.



\bibitem{NardoEtAll} R. Di Nardo, F. Feo, O. Guib\'{e}, Existence result for nonlinear parabolic equations with lower order terms, \textit{Anal. Appl.(Singap.)},  Vol. 09, No. 02,  pp. 161-186, 2011.



\bibitem{Petersen2017}
P. Petersen and  F. Voigtlaender, Optimal approximation of piecewise smooth functions using deep ReLU neural networks, (2017), arXiv: 1709.05289v4.








\bibitem{Pinkus}
A. Pinkus. Approximation theory of the MLP model in neural networks. \textit{Acta Numerica}, Vol. 8, pp. 143–195, 1999.






\bibitem{Porzio} M. M. Porzio, Existence of solutions for some ``noncoercive" parabolic equations, \textit{Discrete and Continuous Dynamical Systems}, Vol. 5, Issue 3,  pp. 553-568, 1999.


\bibitem{Karniadakis1} M. Raissi, P. Perdikaris, and G. Karniadakis. Physics Informed Deep Learning (Part I): Data-driven Solutions of Nonlinear Partial Differential Equations, arXiv:1711.10561, 2017.

\bibitem{Karniadakis2} M. Raissi, P. Perdikaris, and G. Karniadakis. Physics Informed Deep Learning (Part II): Data-driven Discovery of Nonlinear Partial Differential Equations, arXiv:1711.10566, 2017.

\bibitem{Reisinger} C. Reisinger and G. Wittum, Efficient hierarchical approximation of high-dimensional option pricing problems, \textit{SIAM Journal on Scientific Computing}, Vol. 29, No. 1, pgs. 440-458, 2007.

\bibitem{Reisinger2} C. Reisinger, Analysis of linear difference schemes in the sparse grid combination technique, \textit{IMA Journal of Numerical Analysis}, Vol. 33, No. 2, pgs. 544-581, 2012.



\bibitem{ChrisRogers} L.C.G. Rogers, Monte-Carlo Valuation of American Options, \textit{Mathematical Finance}, Vol. 12, No. 3, pgs. 271-286, 2002.

\bibitem{Rudd}  K. Rudd, Solving Partial Differential Equations using Artificial Neural Networks, \textit{PhD Thesis}, Duke University, 2013.

\bibitem{SchmidhuberHighway} R. Srivastava, K. Greff, and J. Schmidhuber, Training very deep networks, \textit{In Advances in Neural Information Processing Systems}, pp. 2377-2385, 2015.


\bibitem{Simon} J. Simon, Compact sets in the space $L^{p}(0,T;B)$, \textit{Annali di Matematica Pura ed Applicata}, Vol. 146, pp. 65-96, 1987.

\bibitem{Tompson}  J. Tompson, K. Schlachter, P. Sprechmann, and K. Perlin.  Accelerating Eulerian Fluid Simulation with Convolutional Networks, \textit{Proceedings of Machine Learning Research}, Vol. 70, pgs. 3424-3433, 2017.


\end{thebibliography}
\end{document}